\documentclass[11pt,english,a4paper,nosumlimits]{article}

\usepackage{amsmath}
\usepackage{amssymb,indentfirst,calc,euscript}
\usepackage[reals]{layout}
\usepackage{xr}
\usepackage{amscd}
\usepackage{setspace}
\usepackage{enumerate}
\usepackage[T1]{fontenc}
\usepackage{sgame}
\usepackage{pstricks}
\usepackage{pst-node,pst-plot}
\usepackage{subfigure}
\usepackage{dsfont}
\usepackage[latin1]{inputenc}
\usepackage{amsfonts}
\usepackage{graphicx}
\usepackage{float}
\usepackage{multirow}
\usepackage{rotating}
\usepackage{stmaryrd}
\usepackage{amsthm}

\RequirePackage{natbib}
\RequirePackage[colorlinks,citecolor=blue,urlcolor=blue,breaklinks=true]{hyperref}
\RequirePackage{hypernat}
\usepackage{breakcites}

\usepackage[normalem]{ulem}

\usepackage{tikz}
\usepackage[reals]{layout}
\usepackage[margin=2.5cm]{geometry}

\newtheorem{theorem}{Theorem}[section]
\newtheorem{proposition}[theorem]{Proposition}
\newtheorem{lemma}[theorem]{Lemma}
\newtheorem{corollary}[theorem]{Corollary}

\newtheorem{definition}[theorem]{Definition}
\newtheorem{claim}[theorem]{Claim}
\newtheorem{fact}[theorem]{Fact}
\newtheorem{remark}[theorem]{Remark}

\newtheorem{notation}[theorem]{Notation}
\newcommand\rv{\boldsymbol }

\newcommand\set{\mathcal }
\newcommand\prob{\mathbb{P}}
\newcommand\reals{\mathbb R}
\newcommand\naturals{\mathbb N}

\let\vare\varepsilon

\let\implies\Rightarrow

\def \cav{ {\rm cav\, } }

\def\freq{ {\rm freq }\, }

\def \argmin{ {\rm argmin\,} }
\def \arg{ {\rm arg} }
\def \argmax{ \mathrm{ argmax\,} }

\def \qed{\hfill $\Box$}%
\def \proof{\noindent {\em Proof.}\,}%

\newcommand{\N}{\mathbb{N}}

\newcommand{\E}{\mathbb{E}}

\font\dsrom=dsrom10 scaled 1200 \def \indic{\textrm{\dsrom{1}}}
\newcommand{\UN}{\indic}

\definecolor{vert}{rgb}{0,0.6,0}

\begin{document}

\title{Persuasion with limited communication capacity\footnote{The authors thank James Best, Olivier Gossner, Frédéric Koessler, Marie Laclau, Daniel Martin, Ludovic Renou, Thomas Rivera, Jakub Steiner, Colin Stewart for stimulating discussions and comments. We thank the editor Alessandro Pavan, an anonymous associate editor and anonymous referees  for helpful comments and suggestions. We also thank participants of the 6th workshop  on Stochastic Methods in Game Theory, Erice May 2017; the 13th European Meeting on Game Theory (SING13), Paris July 2017; the XXVI Colloque Gretsi, Juan-Les-Pins, September 2017; the 10th Transatlantic Theory Workshop, Paris September 2017; the 55th Allerton Conference, Monticello, Illinois, October 2017. We  thank the Institute Henri Poincaré for hosting numerous research meetings.}}

\author{Maël Le Treust\footnote{ETIS UMR 8051, Université Paris Seine, Université Cergy-Pontoise, ENSEA, CNRS, F-95000, Cergy, France; \textsf{mael.le-treust@ensea.fr}; \textsf{sites.google.com/site/maelletreust/}. This research has been conducted as part of the project Labex MME-DII (ANR11-LBX-0023-01). Maël Le Treust gratefully acknowledges financial support from INS2I CNRS, DIM-RFSI, SRV ENSEA, The Paris Seine Initiative and IEA Cergy-Pontoise.}\,  and Tristan Tomala\footnote{HEC Paris and GREGHEC, 1 rue de la Libération, 78351 Jouy-en-Josas, France; \textsf{tomala@hec.fr}; \textsf{sites.google.com/site/tristantomala2}. 
Tristan Tomala gratefully acknowledges the support the HEC foundation and ANR/Investissements d'Avenir under grant ANR-11-IDEX-0003/Labex Ecodec/ANR-11-LABX-0047.}}

\date{Accepted in JET, September 1, 2019}

\maketitle
\begin{abstract}
We consider a Bayesian persuasion problem where the persuader and the decision maker communicate through an imperfect channel that has a fixed and limited number of messages and is subject to exogenous noise. We provide an upper bound on the payoffs the persuader can secure by communicating through the channel. We also show that the bound is tight, i.e., if the persuasion problem consists of a large number of independent copies of the same base problem, then the persuader can achieve this bound arbitrarily closely by using strategies that tie all the problems together. We characterize this optimal payoff as a function of the information-theoretic capacity of the communication channel.

\end{abstract}



\newpage
\section{Introduction}

In modern internet societies, pieces of information are repeatedly and continuously disclosed to decision makers by informed agents. Information transmission is affected by at least two sources of friction. First, the sender and the receiver of a given message may have nonaligned incentives, in which case the sender might be unwilling to transmit truthful information. Second, communication between agents is often imperfect. The sender and the receiver may have time constraints to write or read messages, forcing the sender to summarize his arguments and making him unable to convey all the details. Further, there might be  discrepancies between the informational content of a message that is intended by the sender and the one understood by the receiver. For instance, if the mother tongue of the sender and of the receiver are different, there are possible translation errors \citep[See][]{Blume}. Additionally, messages travelling in a network of computers might be subject to random shocks, internal errors or protocol failures. Studying the effect of noise in communication channels is the starting point of information theory \citep{Shannon48}. 

Our paper aims to study the following questions. How does imperfect communication reduce the possibilities of persuasion in a sender-receiver interaction? When the sender communicates many pieces of information, to what extent does tying the pieces together help in overcoming the communication limitations?\\

We consider a sender and a receiver who communicate over an imperfect channel and are engaged in a series of $n\ge1$ persuasion problems. The sender observes $n$ independent and identically distributed pieces of information and sends $k\ge1$ messages to the receiver. Messages are sent through a channel that consists of two finite sets $X,Y$ of respectively inputs and outputs messages and of a transition probability $Q$ from $X$ to $Y$ such that when the sender chooses input message $x$, the receiver receives output message $y$ with probability $Q(y|x)$. Upon receiving $k$ output messages from the channel, the receiver  chooses $n$ actions, one for each problem. Payoffs are additively separable across persuasion problems. We assume that the sender is able to commit to a disclosure strategy that maps sequences of pieces of information to distributions of sequences of input messages.

We study the optimal average payoff secured by the sender by committing to a strategy. We give an upper bound on this optimal payoff and show that this bound is achieved asymptotically when the numbers $n$ and $k$ grow large.  To prove this latter statement, we borrow techniques from information theory, namely, the coding and decoding schemes of \cite{Shannon48, Shannon59}. This machinery allows to transmit a sequence of messages over a noisy channel with the property that the receiver recovers almost all messages correctly. The information theoretic literature typically considers an obedient receiver who calculates the decoded messages and takes them at face value. In the persuasion game framework, the   receiver is strategic and may not follow any prescribed scheme. Rather, the receiver takes into account the strategy of the sender and the received outputs, calculates its Bayesian belief about the sequence of states, and chooses a sequence of actions that maximizes its payoff. Our technical contribution is to construct a strategy of the sender for which we are able to estimate and to control those Bayesian beliefs in order to ensure that the strategic receiver  chooses a desired sequence of actions.

Our upper bound is the value of an {\em optimal splitting problem with information constraint}, which represents the best payoff that the sender can achieve by sending a message, subject to the constraint that the mutual information between the state and the message is no more than the capacity of the channel. We show that this value is given by the concave closure of the payoff function of the sender, subject to a constraint on the entropy of posterior beliefs. This is also given by the concave closure of a modified payoff function, where the sender pays a cost proportional to the mutual information between the state and the message.\\


\subsection{Motivating example.}
There are relevant situations where a sender discloses information about a large number of independent state parameters. For instance, one can think of testing product quality: a firm has many items to sell, which are ex-ante identical, and the authorities (e.g., the FDA for drugs)  design quality tests \footnote{See e.g., \citealp{perez}.}. One can also think about designing and grading exams to assess the quality of a large number of students\footnote{See \citealp{BoleslavskyCotton} for a model of grading standards through Bayesian persuasion.}.

As an example, consider an innovating firm that has several projects to be financed by investors. The board of investors audits the firm, which is given a limited amount of time to present all the projects.  How to   best structure arguments in order to get the maximum number of projects approved?

To be specific, let us assume that all projects are ex-ante identical and equally likely to be of   good or bad quality. When a project is approved, it yields a positive return of $+1$ to the investors if it is good, and a negative return of $-7$ if it is bad; rejecting a project yields a  payoff of $0$. The objective of the firm is to get a maximum number of projects approved.    

Suppose that the firm commits to an information disclosure mechanism,  i.e.,  distributions of messages conditional on states   \citep[as in][]{KamenicaGentzkow11} and faces no restriction on the number of messages. To invest, the board of investors must be persuaded that the project is good with probability at least $7/8$. Thus, for each project, the firm would optimally draw a good message $g$ or a bad message $b$ with the following probabilities:
$$\prob(g \mid \text{project is good})=1,\hskip0.5cm \prob(g \mid  \text{project is bad})=1/7.$$
This way, the belief that the project is good upon receiving the good message is as follows: $$\prob( \text{project is good}\mid g)=7/8,$$ and the project is accepted with probability $4/7$ (see Section \ref{mainex}).

Now, suppose that the auditing board gives the firm only half the time it would require to talk about all projects. Namely, there is an even number $n$ of projects, but the firm has only $n/2$ messages available. 

A simple strategy the firm can adopt would be to select half of the projects, focus on them, and  communicate optimally for each of them. With this strategy, half of the projects are accepted with probability $4/7$ each, so  in expectation, the average number of accepted projects is $2/7$. This is not optimal, and a better strategy would be to pair projects by two and to draw one message $g,b$ for each pair in the following way:
$$\prob(g \mid \text{both projects  are good})=1,\hskip0.5cm \prob(g \mid  \text{both projects are bad})=0,$$
$$\prob(g \mid  \text{only one project is good})=1/6.$$
The total probability of $g$ is $1/3$ and upon observing this message, the beliefs about quality are as follows:
$$\prob( \text{both projects are good}\mid g)=6/8,$$
$$ \prob( \text{only project 1 is good}\mid g)=\prob( \text{only project 2 is good}\mid g)=1/8.$$
Therefore, each project is believed to be good with probability $7/8$ and both projects are accepted when $g$ is received. Thus, the expected average number of accepted projects is $1/3>2/7$. 

We thus see that tying projects together improves upon communication about each project separately. Suppose that the number of projects is large. Is it possible to find a more complex strategy that further improves the payoff?

Our main result, Theorem \ref{mainthm}, gives an upper bound on the expected average number of accepted projects when the number of messages is half the number of projects. The upper bound is tight: the optimal value approaches it as the number of project increases. In this example, the upper bound is $\lambda^*$ where $(\lambda^*, p^*)$ is the unique solution in $[0,1]\times [0,\frac12]$ of the system of equations:
$$\frac12=\lambda^*\frac78+(1-\lambda^*)p^*, \hskip0.5cm \frac12=\lambda^* H\bigg(\frac78\bigg)+(1-\lambda^*)H(p^*),$$
where $H(p)=-p\log(p)-(1-p)\log(1-p)$ is the entropy function.  The first equation is Bayes plausibility \citep{KamenicaGentzkow11} coming from Bayes' rule, saying that the expected posterior belief is the prior belief. The second equation requires the expected entropy of the posterior to be $\frac12$, which means that the {\em mutual information} between the quality of the project and the message sent to the receiver is equal to the number of messages per project that the firm is able to transmit.

Numerically $\lambda^*\approx 0.519<\frac47\approx 0.571$. Thus, for large $n$, the sender can achieve a payoff better than $1/3$ but bounded away from the payoff obtained with unrestricted communication.


\subsection{Related literature}\label{sec:literature}

We now describe the relationships between our contribution and the literature. This paper is at the junction of Bayesian persuasion and information theory. 

The traditional game theoretic approach to strategic information disclosure assumes perfect communication and analyzes in isolation the problem of sending a single message. These are the well-known sender-receiver games where an informed player, the sender, communicates once with a receiver who takes an action. In the {\em cheap talk} version of this game, the message sent by the sender is costless and unverifiable; see for instance the seminal paper of  \cite{CrawfordSobel1982StrategicInformation}. In the {\em Bayesian persuasion} game \citep{KamenicaGentzkow11}, the sender chooses verifiably an information disclosure device prior to learning his information.  That is, the sender is an {\em information designer}  \citep{Taneva16, BergemannMorris16, BergemannMorris17} who chooses, without knowledge of the state, the information or signaling structure which releases information to the decision maker. 

In parallel, information theory considers agents with perfectly aligned interests and analyzes the {\em rate} of information transmission.  The sender observes an information {\em flow}, which is a stochastic process, and sends messages to the receiver over an imperfect channel represented by a transition probability from input to output messages. Truthful information transmission is the common goal of the sender and the receiver. The {\em rate of information transmission} is the average number of correct guesses made by the receiver. 
Shannon's theory \citep{Shannon48, Shannon59} determines whether a source of information can be transmitted over the channel with arbitrarily small probability of error and shows that the rate of the source of information has to be smaller than {\em the capacity of the channel} defined as the maximal {\em mutual information} between input and output messages.\\

Our model of persuasion has two essential features. The sender and the receiver are engaged in a large number of identical copies of the same game and communication is restricted to an imperfect channel. As \cite{KamenicaGentzkow11}, we consider the payoff obtained by the sender as a function of the belief of the receiver, when the receiver takes optimal actions. 
With unrestricted communication, that is on a perfect channel with large set of inputs, the optimal payoff for the sender is given by the concave closure of this function. Then, solving any number of identical games amounts to solving each copy separately.  With a single copy, the game of persuasion with a noisy channel is studied by \cite{TsakasTsakas2017} who prove the existence of optimal solutions and show monotonicity of the sender's payoff with respect to the noise of the channel.
Considering many copies of the base game {\em and\,} restricted communication, we show that {\em linking independent problems together} yields a better payoff to the sender: the optimal strategy correlates all messages with the state parameters of all problems. In this respect, our work bears some similarity with \cite{JacksonSonnenschein07}, who showed that a mechanism designer can achieve more outcomes in an incentive compatible manner by linking many identical problems together. \\

The optimal payoff that we characterize is related to models where the cost of information is measured by mutual information. Such information costs have been introduced in the literature on rational inattention by \cite{Sims03}, \cite[See also][]{MatejkaMcKay, Martin17, MatejkaSteinerStewart}. The use of mutual information has been axiomatized in \cite{MorrisStrack} and \cite{HebertWoodford}. In the context of persuasion, \cite{GentzkowKamenica14} consider a model where the sender gets his payoff from the game, minus a cost that is proportional to the mutual information between the state and the message; see also \cite{Matyskova}. With Lagrangian methods, we find that the value of our optimal splitting problem with information constraint is the concave closure of the payoff function, net of such an information cost, a similar concavification problem is found in \cite{CaplinDean}

Different from those papers, the mutual information is not a primitive of our model. Our finding is that the noise and limitations in communication induce a {\em shadow cost} measured by the mutual information.\\

Entropy and mutual information appear endogenously in several papers on repeated games 
where players have bounded rationality \citep{NeymanOkada99, NeymanOkada00}, are not able to freely randomize their actions \citep{GossnerVieille02}, or observe actions imperfectly \citep{GossnerTomala06, GossnerTomala07}.  A related paper is \cite{GossnerHernandezNeyman06}, henceforth GHN, who also consider a sender-receiver game. In GHN, the sender and the receiver play an infinitely repeated game with common interests: both the sender and the receiver want to choose the action that matches the state. The sender knows the infinite sequence of states and can communicate with the receiver only through his actions. GHN characterize the best average payoff that the sender (and the receiver) can achieve. Their solution resembles ours: the optimal value is the payoff obtained when the sender can send a direct message to the receiver, subject to an information constraint. 

There are important differences with our work. First, GHN study a cheap talk game with common interests. By contrast, we do not assume common interests and we assume commitment power for the sender. Second, GHN is a truly repeated game model: at any given time $t$, both players choose actions and the information of the receiver at this time consists of past actions. In our case, the sender knows a finite sequence of states and chooses a finite sequence of input messages, the receiver observes a finite sequence of output messages and chooses a sequence of actions. This is why, rather than seeing our model as a repeated game of persuasion, we view it as a {\em spatial} model with identical copies of the same problem coexisting at the same time.  This also explains why the number of copies $n$ need not be equal to the number of times $k$ the channel is used by the sender. Our result characterizes the optimal payoff as a function of the ratio of the number $n$ of pieces of information to the number $k$ of channel uses. In particular, this allows us to analyze cases where the channel is perfect (i.e. not subject to random noise) but with limited input size: there are fewer messages than states or actions.

Cheap talk with a noisy channel has been studied by \cite{Blume} who show that the presence of noise is possibly welfare improving. Such a phenomenon cannot happen in the persuasion context as the sender could commit to replicate the noise. Relatedly, \cite{HernandezVonStengel14} consider a sender-receiver game with common interests over an imperfect channel. In that paper, there is only one state known by the sender and one action taken by the receiver, while the channel can be used a fixed number of times.  \cite{HernandezVonStengel14} characterize all the Nash equilibria of this game and study the differences with Shannon's coding methods.  Again, we do not assume common interests and   assume commitment power for the sender.  More importantly, our focus is different and more in line with GHN: we do not treat a single persuasion problem but a large sequence of them and use information theory to study the asymptotics of the problem. \\

Our work is also related to some information theoretic literature. Following GHN, a line of papers study {\em empirical coordination} between a sender and a receiver \citep{CuffPermuterCover10, Cuff(ImplicitCoordination)11, LeTreust(EmpiricalCoordination)17}. Assuming common interest between the sender and the receiver, those papers characterize the asymptotic empirical distributions of (states, messages, actions) which are achievable, given the information structure and the noisy channel. 
The closest paper in this literature is \cite{LeTreustTomala(Allerton)16} where we have
studied empirical coordination between a persuader and a decision maker induced by approximate equilibria as the number of repetitions tends to infinity. Recently, \cite{AkyolLangbortBasarIEEE17} have considered the problem of Bayesian persuasion in a model with Gaussian states and channel and quadratic functions as in \cite{CrawfordSobel1982StrategicInformation}.
\\

The remainder of this paper is organized as follows. The model is described in Section \ref{sec:model} and we state our main results in Section \ref{sec:main}. In Section \ref{mainex}, we illustrate our results with a detailed example.
We provide an extension in Section \ref{sec:extension} and concluding comments in Section \ref{sec:conclusion}. Proofs are in the Appendix.

\section{Model}\label{sec:model}

\subsection{The persuasion problem} 
In this model, we consider a sender ($S$) and a receiver ($R$) engaged in a series of identical persuasion problems and where the communication technology is fixed exogenously.

There is a finite state space $\Omega$ endowed with a prior probability distribution $\mu$, a finite action set $A$ for the receiver, and each player $i=S,R$ has a payoff function $u_i:\Omega\times A\to \reals$. There is also a fixed communication channel $(X,Y,Q)$, where $X,Y$ are finite sets of messages and $Q:X\to \Delta(Y)$ is a transition probability from $X$ to $Y$ (henceforth $\Delta(S)$ denotes the set of probability distributions over the finite set $S$).

Given two integers $n,k$, we define a repeated persuasion problem where the uncertainty is about a sequence  $\omega^n=(\omega_1,\dots,\omega_n)$ drawn i.i.d. from $(\Omega,\mu)$. The receiver chooses a sequence of actions $a^n=(a_1,\dots, a_n)$ and the payoff for player $i=S,R$ is as follows:
$$\bar u_i(\omega^n,a^n)=\frac1n\sum_{t=1}^nu_i(\omega_t,a_t).$$
To disclose information, the sender can use the channel $k$ times by choosing a sequence of input messages $x^k=(x_1,\dots, x_k)$. The channel then draws a sequence of output messages $y^k$ with probability $Q^k(y^k|x^k)=\prod_{t=1}^kQ(y_t|x_t)$ and sends it to the receiver.

This defines the following  persuasion game $\Gamma(n,k)$:
\begin{enumerate}
\item  The sender chooses a strategy $\sigma:\Omega^n\to \Delta(X^k)$ which is announced to the receiver.
\item A sequence of states $\omega^n$ is drawn  i.i.d.~from the prior $\mu$, a sequence of  input messages $x^k$ is drawn with probability $\sigma(x^k|\omega^n)$, a sequence of  output messages $y^k$ is drawn with probability $Q^k(y^k|x^k)$  and is observed by the receiver.
\item The receiver chooses a sequence of actions $a^n$.
\end{enumerate}
Then, player $i=S,R$ gets the average payoff $\bar u_i(\omega^n,a^n)$.\\

Notice that for $n=k=1$, this is the model of \cite{TsakasTsakas2017} of a single persuasion problem with noisy communication. An interesting particular case is given by {\em perfect} channels where $X=Y$ and $Q(y|x)=\indic_{\{y=x\}}$. In such a case, the only limitation is given by the number of messages. If we let $n= k=1$ and choose a perfect channel with sufficiently many messages $|X|=|Y|\geq |\Omega|$, the model encompasses the standard persuasion game of \cite{KamenicaGentzkow11}.

\subsection{Optimal robust payoff} 
As a solution concept, we study the best payoff the sender can secure, regardless of which best reply is chosen by the receiver.  A strategy of the receiver is a mapping $\tau:Y^k\to A^n$. Knowing $\sigma$, the receiver chooses a best reply $\tau$, which maximizes the expected payoff. That is, for each $y^k$: 
$$\tau(y^k)\in\underset{a^n\in A^n}\argmax\sum_{\omega^n,x^k}\mu^n(\omega^n)\sigma(x^k|\omega^n)Q(y^k|x^k)\bar{u}_R(\omega^n,a^n).$$
Denote $BR(\sigma)$ the set of best replies of the receiver to the strategy $\sigma$.

\begin{definition}
The {\em optimal robust payoff} of the sender in this problem is as follows:
$$U^*_S(\mu^n,Q^k)=\sup_\sigma\min_{\tau\in BR(\sigma)}\sum_{\omega^n,x^k,y^k}\mu^n(\omega^n)\sigma(x^k|\omega^n)Q^k(y^k|x^k)\bar u_S(\omega^n,\tau(y^k)).$$
\end{definition}

This definition differs from the conventional solution to Bayesian persuasion of \cite{KamenicaGentzkow11} where the receiver takes the best reply which is preferred by the sender. Our choice is motivated by robustness; we ask the solution to be robust to the way the receiver breaks ties\footnote{A similar approach is followed by \cite{pavan} and \cite{mathevet}.}. We stress that this choice does not matter for generic problems. Indeed, with slight perturbations of the payoff function of the receiver, we can make sure that indifferences occur only at interior beliefs. When this is the case, the sender can slightly change his strategy in order to avoid the indifference region.

The goal of this paper is to give an upper bound for the optimal robust payoff and to characterize its limit when $n$ and $k$ tend to infinity.

\subsection{Optimal splitting problem with information constraint} 
To state our main results, we introduce some definitions.
\begin{definition} 
A splitting of $\mu\in \Delta(\Omega)$ is a finite family $(\lambda_m ,\nu_m)_m$, where for each $m$, $\nu_m\in \Delta(\Omega)$, $\lambda_m\in[0,1]$, $\sum_m\lambda_m = 1$ such that:
\begin{align}
\mu=&\sum_m\lambda_m \nu_m.
\end{align}
\end{definition}

A splitting of $\mu$ is a distribution of posterior beliefs whose average equals the prior. An ``information structure'' which draws a message $m$ with probability $\prob(m|\omega)$ in state $\omega$, induces a splitting $(\lambda_m,\nu_m)_m$ with $\lambda_m=\sum_{\omega'}\mu(\omega')\prob(m|\omega')$  and $\nu_m(\omega)=\frac{\mu(\omega)\prob(m|\omega)}{\sum_{\omega'}\mu(\omega')\prob(m|\omega')}$. From the splitting lemma 
\citep{AM95} or Bayes plausibility \citep{KamenicaGentzkow11}, for each decomposition of the prior belief into a convex combination of posterior  $\mu=\sum_m\lambda_m \nu_m$, the splitting $(\lambda_m,\nu_m)_m$ is induced by some information structure, for example, $\prob(m|\omega)=\lambda_m\nu_m(\omega)/\mu(\omega)$.

For each posterior belief $\nu\in\Delta(\Omega)$, let the set of optimal actions of the receiver be:
$$A^*(\nu)=\underset{a\in A}\argmax\sum_\omega \nu(\omega)u_{R}(\omega,a).$$
We denote by $u^*_S(\nu)=\min_{a\in A^*(\nu)}\sum_\omega \nu(\omega)u_{S}(\omega,a)$  the {\em robust payoff} of the sender at the belief $\nu$, i.e., the payoff of the sender when the receiver chooses the optimal action, which is {\em worst} for $S$.

We now introduce tools borrowed from information theory; the reader is referred to \cite{cover-book-2006}. 
\begin{definition} 
\begin{enumerate}
\item The (Shannon) entropy of a probability distribution $q\in\Delta(S)$ over a finite set $S$ is as follows:
$$H(q) = -\sum_{q} q(s)\log q(s),$$
where the logarithm has basis $2$ and $0\log0=0$. 
\item The mutual information between two random variables $(\rv x,\rv y)$, drawn from the joint probability distribution $p(x) Q(y|x)$ is as follows:
 $$ I_{p,Q}(\rv x;\rv y)=H\Big(\sum_x p(x)Q(\cdot |x)\Big)-\sum_x p(x)H(Q(\cdot |x))$$

\item The capacity of the channel $(X,Y,Q)$ is as follows:
$$C(Q)=\max_{p\in\Delta(X)}   I_{p,Q}(\rv x;\rv y).$$

\end{enumerate}
\end{definition}

The channel capacity $C(Q)$ is the maximal mutual information between two random variables $(\rv x,\rv y)$, respectively the input and output of the channel, drawn from the joint probability distribution $p(x) Q(y|x)$, where the maximum is over the marginal distribution $p(x)$. 
Intuitively, this is the maximal number of bits of information that can be transmitted reliably through the channel (see \citealp{cover-book-2006}).

Equipped with these tools, our main definition is the following.
\begin{definition} For any $c\ge0$, the optimal splitting problem with information constraint is: \begin{eqnarray*}
V(\mu,c)=&\sup & \sum_m\lambda_m u^*_S(\nu_m)\\
&\mathrm{ s.t.} &  \sum_m\lambda_m\nu_m=\mu,\\
&\mathrm{ and } &  H(\mu)-\sum_m\lambda_m H(\nu_m)\leq c.
\end{eqnarray*}
\end{definition}
This is the best payoff that the sender can secure by choosing a splitting of the prior belief (i.e., an information structure) under the constraint that the expected reduction of entropy does not exceed the capacity $c$ of the channel. The entropy reduction $ H(\mu)-\sum_m\lambda_m H(\nu_m)$ is nonnegative and is the mutual information between a random state $\rv \omega$ and a random message $\rv m$, drawn from the joint distribution  $\big(\lambda_m\nu_m(\omega)\big)_{(\omega,m)}$. The interpretation is thus that the sender optimizes over a set of information structures that convey bounded information about the state.

Notice that $V(\mu,c)$ is less than or equal to the concave closure (or concavification) of $u^*_S$ at $\mu$ which is the unconstrained supremum $ \cav u^*_S(\mu):=\sup\big\{\sum_m\lambda_m u^*_S(\nu_m) :  \sum_m\lambda_m\nu_m=\mu\big\}.$

\section{Results}\label{sec:main}
\subsection{The main result}
The main result of this paper shows that the value of the optimal splitting problem with information constraint provides an upper bound to the optimal robust payoff and that the bound is achieved asymptotically.
\begin{theorem}\label{mainthm}
\begin{enumerate}
\item The optimal robust payoff of the sender is no more than the value of the optimal splitting problem with information constraint. For each pair of integers $n,k$:
$$U^*_S(\mu^n,Q^k)\leq V\Big(\mu,\frac{k}{n}C(Q)\Big).$$

\item  The optimal robust payoff of the sender converges to the value of the optimal splitting problem with information constraint in the following sense. For each $r\in[0,+\infty]$, for each pair of sequences of integers $(k_j,n_j)_{j\in\naturals}$ such that $\lim\limits_{j\to\infty}\max(n_j,k_j)= \infty$ and 
 $\lim\limits_{j\to\infty}\frac{k_j}{n_j}=r$, we have:
$$\lim_{j\to\infty}U^*_S(\mu^{n_j},Q^{k_j}) = V\big(\mu,r C(Q)\big).$$

\end{enumerate}
\end{theorem}

On the one hand, this result shows communication restrictions limits the payoff that can be achieved through Bayesian persuasion. On the other hand, it quantifies the extent to which repeating the same problem and linking the copies together helps in overcoming those restrictions.

\subsubsection{Sketch of proof} We give an intuition for the main arguments of the  proof; the technical details are in the appendix.
 
\paragraph{First point, upper bound.} The argument is that regardless of which strategies are used, the mutual information between the states and the messages to the receiver cannot exceed the capacity of the channel. 

For simplicity, consider the case $n=k=1$ where the result says  $U^*_S(\mu, Q)\leq V(\mu, C(Q))$. Take any strategy $\sigma$ of the sender. This induces the splitting
$\mu=\sum_y\prob_\sigma(y)\nu_y$ where $\prob_\sigma(y)=\sum_{\omega,x}\mu(\omega)\sigma(x|\omega)Q(y|x)$ is the probability of the message $y$ and $$\nu_y(\omega)=\prob_\sigma(\omega|y)=\frac{\sum_x\mu(\omega)\sigma(x|\omega)Q(y|x)}{\prob_\sigma(y)}$$
is the posterior belief conditional on $y$.  The mutual information of this splitting is:
$$H(\mu)-\sum_y\prob_\sigma(y) H(\nu_y):=I(\rv\omega;\rv y)$$
where $(\rv \omega,\rv x,\rv y)$ denotes a random triple of state, input and output messages drawn from the joint distribution $\mu(\omega)\sigma(x|\omega)Q(y|x)$. With an abuse of notation, we denote $I(\rv\omega;\rv y)$ the mutual information between $\rv \omega$ and $\rv y$ without explicit reference to the distribution.

Since $\rv x$ is a sufficient statistic for $\rv y$,  $\rv x$ is more informative\footnote{\citealp[See][Theorem 2.8.1, p. 34]{cover-book-2006}.} about $\rv y$ than $\rv \omega$, that is $I(\rv\omega;\rv y)\leq I(\rv x;\rv y)$. Then, the mutual information between the input and the output is no more than $C(Q)$ from the definition of the channel capacity.

The proof for general $n$ and $k$ is an elaboration of this argument. Since states are i.i.d., we can prove that for any strategy, the average payoff is the one induced by some splitting whose mutual information is no more than $\frac{k}{n}C(Q)$. 
The trick is to introduce an auxiliary random variable $\rv t$ uniformly distributed over $\{1,\dots, n\}$ and known by the receiver. Then, we regard the average payoff over stages $1,\dots, n$ as the expected payoff for the randomly selected stage.

\paragraph{Second point, asymptotic construction.}  To make the intuition simple, let us consider a sequence of pairs of integers $(k_j,n_j)_{j\in \naturals}$ such that $k_j=n_j$ and let $k=n$ be a large term of this sequence.
Take a splitting $(\lambda_m,\nu_m)_m$ of the prior $\mu$ which satisfies the information constraint. We want to show that for large $n$, there is a strategy $\sigma$ of the sender such that for any best reply $\tau\in BR(\sigma)$ of the receiver, the payoff of the sender is at least about $\sum_m\lambda_m u^*_S(\nu_m)$. Let also $a^*_m\in A^*(\nu_m)$ such that $u^*_S(\nu_m)=\sum_\omega\nu_m(\omega)u_{S}(\omega,a^*_m)$.

A first intuition for the construction is as follows. From Shannon's coding Theorem\footnote{\citealp[See][Theorem 10.4.1, p. 318]{cover-book-2006}.}, if $I(\rv \omega;\rv m)< C(Q)$,  then for large $n$, there exists  functions $f_1:\Omega^n\to M^n$, $f_2:M^n\to X^n$ and $g:Y^n\to M^n$, altogether a coding/decoding scheme, with the following properties. Given a sequence of states $\omega^n$, the sender calculates a sequence of messages $m^n=f_1(\omega^n)$ such that with probability close to one, the empirical frequency of the  $(\omega_t,m_t)$'s is approximately the theoretical one $\lambda_m\nu_m(\omega)$. The sender then calculates a sequence of inputs $x^n=f_2(m^n)$ and sends them into the channel. If the receiver calculates $\hat m^n=g(y^n)$, then the messages are recovered with probability close to one: $\prob(m^n= \hat m^n)\approx 1$.

This argument is standard in information theory but is not sufficient for proving our result. The proof is actually more complicated because the strategic receiver actually calculates the Bayesian posterior $\prob(\omega^n|y^n)$ and chooses at stage $t$ an action $a_t\in A^*(\prob(\omega_t|y^n))$. Thus, the main task is to refine the construction in such a way that for any best reply of the receiver, with probability close to one, the optimal action $a_t\in A^*(\prob(\omega_t|y^n))$ is equal to the recommended action $a^*_{\hat m_t}$ at most stages, that is, for a set of stages whose proportion is close to one. This implies that the payoff is approximately the target one. 

The proof consists of three main steps. In the first step, we show that for each $\vare>0$, we can find a splitting $\vare$-optimal for $V(\mu,c)$, which satisfies the information constraint with strict inequality and such that for each posterior $\nu_m$, the action $a^*_m$ which minimizes the sender payoff over $A^*(\nu_m)$ is \emph{unique} in a neighborhood of $\nu_m$. This latter property ensures that the receiver plays $a^*_m$ whenever its belief is close to  $\nu_m$. We deduce that the difference between the realized payoff and the target payoff is bounded by the number of times $t$ where the Bayesian posterior $\prob(\omega_t|y^n)$ is far away from $\nu_{\hat m_t}$. The goal is then to show that this number is small with probability close to one.

The second step consists in defining Shannon's strategy for this splitting. There, we adapt known construction from information theory to our setting.  
 
At the third step, we prove that, under our construction, with probability close to one, the Bayesian posteriors $\prob(\omega_t|y^n)$ are close enough to the target posteriors $\nu_{m_t}$ at most stages. This allows us to conclude that with probability close to one, the receiver plays the recommended actions at most stages and that the expected payoff is close to the target one. This step, where we estimate the realized Bayesian beliefs, is new compared to the information theoretic literature, which typically focuses on the average number of mistakes in decoding. Summing up, our construction is similar to the ones found in this literature but is adapted to the context where the receiver is maximizing its payoff.
\hfill $\Box$

\subsubsection{Implications}\label{implications}

We now provide some direct implications of the theorem.
\paragraph{Large capacity.}  Reordering the  information constraint as $\sum_m\lambda_mH(\nu_m)\geq H(\mu)-c$, we see that if $c\geq H(\mu)$, the constraint is satisfied by all splittings. The value of the problem is thus the unconstrained concavification of $u^*_S$: 
$$c\geq H(\mu)\implies V(\mu,c)=\cav u^*_S(\mu).$$ 
As a consequence, if we fix $n$ and $Q$ and choose $k$ large enough such that $\frac{k}{n}C(Q)\geq H(\mu)$, then the sender can achieve approximately the unconstrained maximum $\cav u^*_S(\mu)$.

The intuition is simple: for fixed size of the state space, if the imperfect channel can be used a large number of times, then the sender is able to convey any message with arbitrarily high probability. More precisely, suppose $C(Q)>0$ that is to say, $Q(\cdot|x)$ is not constant with respect to $x$. There exist distributions of inputs $p_m\in\Delta(X)$ that statistically identify the message:
$$m\neq m'\implies \sum_x p_m(x)Q(\cdot|x)\neq \sum_x p_{m'}(x)Q(\cdot|x).$$
For each message $m$, the sender can draw an i.i.d. sequence of messages $x_1,\dots, x_k$ from $p_m$ and sends them through the channel. The posterior belief of the receiver conditional on $y_1,\dots, y_k$ then converges to the truth (the Dirac mass on $m$). Thus, asymptotically, the distributions of actions of the receiver will be close to the one under perfect communication.

\paragraph{Small capacity.} When $c$ is close to $0$, the information constraint
$H(\mu)-\sum_m\lambda_m H(\nu_m)\leq c$ implies that the splitting is almost nonrevealing since:\footnote{See \citealp[Lemma 11.6.1, p. 370]{cover-book-2006}.}
$$\sum_m\lambda_m\|\nu_m-\mu\|_1\leq \sqrt{2\ln 2\, \big(H(\mu)-\sum_m\lambda_m H(\nu_m)\big)}.$$ 
It follows that $V(\mu,c)$ is approximately $u^*_S(\mu)$, the payoff obtained without any information transmission.

As a consequence, if we fix $Q$ and $k$, then for large $n$, the sender cannot get substantially more than $u^*_S(\mu)$.

\paragraph{Perfect channels.}  Our result applies to communication channels without noise. A communication channel has two sources of imperfection: the noise  and the number of available messages, which is given exogenously. One insight of our work is that all that matters for the analysis is the capacity of the channel. 

A channel $(X,Y,Q)$ is called {\em perfect} if $X=Y$ and $Q(y|x)=\indic_{\{x=y\}}$. For each integer $m\geq 2$, we denote $Q^*_m$ the perfect communication channel with $m$ messages where $m=|X|=|Y|$.  Its capacity is\footnote{See \citealp[p. 184]{cover-book-2006}.}  $C(Q^*_m)=\log m.$ We apply our results to the optimal robust payoff $U^*_S(\mu^n,Q^*_{m})$ of the game where the persuasion problem is repeated $n$ times and where the sender can send {\em one} message from a set with cardinality $m$. Our method applies since for large $m$, the channel $Q^*_m$ can be seen as having the use of a binary perfect channel $k$ times, with $k=\log_2 m$.

There are two simple extreme cases. First, if $m=1$, the capacity of the channel is 0 and the sender cannot convey any information. Thus, $U^*_S(\mu^n,Q^*_{m})=V(\mu,0)=u^*_S(\mu).$
Second, if $m\ge|\Omega|^n$,  then the sender can secure the unconstrained persuasion payoff $U^*_S(\mu^n,Q^*_{m})= V(\mu,\log|\Omega|)=\cav u^*_S(\mu)$  by treating each of the $n$ problems separately and getting the payoff $\cav u^*_S(\mu)$ for each instance. The first point of Theorem \ref{mainthm} shows that this is the best possible payoff.

More generally, Theorem \ref{mainthm} implies the following.

\begin{corollary} Consider a persuasion problem repeated $n$ times, where the sender sends one message from a set of cardinality $m$. Then:
\begin{enumerate}
\item $U^*_S(\mu^n,Q^*_{m})\leq V(\mu,\frac{\log m}{n})$.
\item  For any pair of sequences of integers $(m_j,n_j)_{j\in\naturals}$ such that $\lim\limits_{j \to \infty}\max(m_j,n_j)=\infty$ and $\lim\limits_{j \to \infty}\frac{\log m_j}{n_j}= c$,  we have $\lim\limits_{j \to \infty}U^*_S(\mu^{n_j},Q^*_{m_j}) = V(\mu,c)$.
\end{enumerate}
\end{corollary}

\begin{proof} The first point follows directly from Theorem \ref{mainthm}. To see the second point, it is enough to remark that a perfect channel $ Q^*_m$ is ``close'' to $k$ copies of a perfect binary channel with  $k$  such that $2^{k}\leq m<2^{k+1}$, 
that is $k=\lfloor \log m\rfloor$. Having more messages at disposal is beneficial for the sender and thus $U^*_S(\mu^n,Q^*_{m})$ is weakly increasing with $m$. It follows that:
$$U^*_S(\mu^n, (Q^*_{2})^k)\leq U^*_S(\mu^n, Q^*_{m})\leq U^*_S(\mu^n, (Q^*_{2})^{k+1}).$$
 Take a sequence $(m_j,n_j)_{j\in\naturals}$ such that $\lim\limits_{j \to \infty}\max(m_j,n_j)=\infty$ and $\lim\limits_{j \to \infty}\frac{\log m_j}{n_j}= c$, and define $k_j=\lfloor \log m_j\rfloor$. We have $\lim\limits_{j \to \infty}\max(k_j,n_j)=\infty$, $\lim\limits_{j \to \infty}\frac{k_j}{n_j}= c$ and the conclusion follows from Theorem  \ref{mainthm}.
\end{proof}

\subsection{Concavification with information constraint}\label{sec:cav}
In this section, we give some properties of the optimal splitting problem under information constraint.
The motivation for this part of the results is two-fold. First, it is known than in a concavification problem, the number of posteriors (or of messages) can be chosen less than or equal to the number of states. One might wonder whether this remains true when there is a constraint on the feasible splittings. Second, models with costly information often use the mutual information as information cost  \citep[See e.g.][]{Sims03}. We will see that in our case, this is derived by writing a Lagrangian for $V(\mu,c)$.

Consider the optimal splitting under information constraint:
$$\sup\Big\{\sum_m\lambda_m u^*_S(\nu_m) : \sum_m\lambda_m \nu_m=\mu, \sum_m\lambda_mH(\nu_m)\geq H(\mu)-c\Big\}.$$
This is a special instance of the following optimization problem. Let $f,g:X\to\reals\cup\{-\infty\}$ be two functions defined on a convex set $X\subseteq \reals^d$, where $X$ represents an abstract set of posteriors, $f$  is a payoff function and $g$ is a constraint capturing the feasible splittings. For $x\in X$ and $\gamma\in \reals$ consider the  problem:
$$F^g(x,\gamma):=\sup\Big\{\sum_m\lambda_m f(x_m) : \sum_m\lambda_m x_m=x, 
\sum_m\lambda_mg(x_m)\geq \gamma\Big\}.$$

Let $f^g:X\times \reals\to \reals\cup\{-\infty\}$ defined by:
\[ f^g(x,\gamma)=
\begin{cases}
f(x) & \text{ if } \gamma\leq g(x), \\
-\infty & \text{ otherwise. } \\
\end{cases}
\]

\begin{theorem}\label{thm:cavcontraint}
Then, for each $(x,\gamma)\in X\times \reals$,
\begin{enumerate}
\item
$F^g(x,\gamma)=\cav f^g(x,\gamma).\label{eq:cav}$

\item $F^g(x,\gamma)=\inf_{t\geq 0}\Big\{\cav(f+tg)(x)-t\gamma\Big\}.$

\end{enumerate}
\end{theorem}

Applying this result to the optimal splitting under information constraint, we get:

\begin{corollary}\label{coro:Cav} For each $\mu\in\Delta(\Omega)$ and $c\ge0$, 
\begin{enumerate}
\item $V(\mu, c)$ is the concavification of the function $u^H_S:\Delta(\Omega)\times\reals\to\reals$ defined as:
\[ u^H_S(\nu,\eta)=
\begin{cases}
u^*_S(\nu) & \text{ if } \eta\leq H(\nu), \\
-\infty & \text{ otherwise, } \\
\end{cases}
\]
calculated at $(\nu,\eta)=(\mu, H(\mu)-c)$.

\item  
$V(\mu,c)=\inf_{t\ge0}\Big\{\cav(u^*_S+tH)(\mu)-t(H(\mu)-c)\Big\}.$\\

\end{enumerate}
\end{corollary}

Since it might be useful in other contexts, Theorem \ref{thm:cavcontraint} is stated for general functions rather than specifically for the entropy function. This result has recently been generalized by \cite{DovalSkreta} to splitting problems with several constraints.
The first point of the theorem states that the concavification with constraint, is the concavification of a bivariate function where an additional variable is added for the constraint (many variables when there are many constraints, see \citealp{DovalSkreta}). 
The second point states that a Lagrangian function can be introduced and that the concavification under constraint is the concavification of the Lagrangian for some multiplier. 
The proof is in the Appendix (\ref{A21}).

A direct implication of the second point of Corollary \ref{coro:Cav} is that there exists $t^*=t^*(\mu,c)$ such that:
 $$V(\mu,c)=\cav(u^*_S+t^*H)(\mu)-t^*(H(\mu)-c).$$ 
To see the existence of $t^*$, notice that $\cav(u^*_S+tH)(\mu)-t(H(\mu)-c)\geq (u^*_S+tH)(\mu)-t(H(\mu)-c)=u^*_S(\mu)+tc,$ which tends to $+\infty$ as $t\to+\infty$. Therefore, $t\mapsto\cav(u^*_S+tH)(\mu)-t(H(\mu)-c)$ reaches a minimum at some $t^*$.

If $(\lambda^*_m,\nu^*_m)_m$ is an optimal splitting, let $\set I^*= H(\mu)-\sum_m\lambda^*_mH(\nu_m^*)$ be its mutual information. We have the following:
\begin{align}\label{eq:lagrange}V(\mu,c)=\sum_m\lambda^*_mu^*_S(\nu_m^*)-t^*(\set I^*-c).\end{align}
We then find the usual Kuhn-Tucker slackness conditions. If $\set I^*<c$, then $t^*=0$ and the unconstrained optimum  is feasible. 
If $t^*>0$, the constraint is binding.  The Lagrange multiplier $t^*$ can be interpreted as the  {\em shadow price of capacity}, that is, the marginal value of an extra unit of communication capacity. 

This characterization can be related with the {\em cost of information} considered in the literature on rational inattention \citep[See][]{Sims03} where the agent pays a cost proportional to the mutual information between the state and the signal he observes. In particular, \cite{CaplinDean} consider the concavification of a utility function net of such an information cost. For persuasion games, \cite{GentzkowKamenica14} assume that the sender pays a cost for choosing a disclosure strategy which is also related to the mutual information and also take the concavification of the net utility function.

Equation (\ref{eq:lagrange}) can be seen as a microfoundation of the use of mutual information as the information cost: the  limit optimal value of persuasion for a large number of copies of problems with communication over an imperfect channel, has the same value as a problem of persuasion with an information cost. 
There are some differences, however. First, the information cost is not the mutual information, but the difference between the mutual information and the capacity of the channel. That is, a cost reduces the payoff only when the sender would like to send more information bits than the capacity. Second, the unit price of capacity is endogenous and given by the Lagrange multiplier of the information constraint.\\

A direct implication is an upper bound of the number of posteriors needed to achieve the concavification.

\begin{corollary}\label{lem:actions} In the optimization problem,
$$V(\mu,c)=\sup\Big\{\sum_m\lambda_m u^*_S(\nu_m) : \sum_m\lambda_m\nu_m=\mu,\sum_m\lambda_mH(\nu_m)\geq H(\mu)-c \Big\},$$
the number of posteriors can be chosen to be at most $\min\{|A|, |\Omega|+1\}$.
\end{corollary}

Without the information constraint, the usual bound is $\min\{|A|, |\Omega|\}$: the number of posteriors or of messages can be upper bounded by the number of actions and the number of states. For the number of actions, the argument is that two messages for which the receiver chooses the same action can be merged into one and the corresponding two posteriors replaced by the average. The argument still holds due to the concavity of the entropy function: replacing two posteriors by their average increases the expected entropy and thus helps in satisfying the information constraint.

For the bound given by the number of states, the usual technical argument is that any point in the convex hull of the hypograph of a function on $\Delta(\Omega)$ is a convex combination involving $|\Omega|$ points. From Corollary \ref{coro:Cav}, we consider the concavification of a function defined on $\Delta(\Omega)\times \reals$ a domain with one extra dimension; thus, an extra posterior might be needed. A similar observation is made in \cite{BoleslavskyKim},
where due to an incentive constraint, an extra posterior is needed. In Section \ref{mainex}, we provide an example where $|\Omega|+1$ posteriors are used at the optimum.

\section{Illustrating example}\label{mainex}

\subsection{Unrestricted communication}
In this example, the sender is a firm that persuades the receiver to invest in a risky project. If the receiver does not invest (action $a_0$), the payoff is 0 for both players. If the receiver invests (action $a_1$), the project has return $-7$ in the bad state $\omega_0$ and  $+1$ in the good state $\omega_1$. Both states are equally likely. The sender receives a fee of $+1$ only if the receiver invests. The payoff table is as follows, the entries are pairs of payoffs for the players $i=S,R$ depending on the state and action.

\begin{center}
\begin{tabular}{cccc}
                            & $a_0$                   & $a_1$                      &  $\mu$   \\ \cline{2-3}
\multicolumn{1}{l|}{$\omega_0$} & \multicolumn{1}{l|}{$\;\,0,0\,\;$} & \multicolumn{1}{l|}{$\,1,-7$} &  $\frac12$  \\ \cline{2-3}
\multicolumn{1}{l|}{$\omega_1$}   & \multicolumn{1}{l|}{$\;\,0,0\,\;$} & \multicolumn{1}{l|}{$\;\,1,1$} &  $\frac12$  \\ \cline{2-3}
                            &                          &                           &  
\end{tabular}
\end{center}

The receiver invests for sure only when he holds a belief $\nu$ such that $\nu(\omega_1)>7/8$. If $\nu(\omega_1)=7/8$ he is indifferent. Assuming that in case of indifference he does not invest, the robust payoff of the sender $u^*_S(\nu)$ is $1$ if $\nu(\omega_1)>7/8$ and 0 otherwise.

\begin{figure}[!ht]
\begin{center}
\psset{xunit=5cm,yunit=5cm}
\begin{pspicture}(0,-0.2)(1,1.3)
\psline{->}(0,-0.1)(0,1.2)
\psline{->}(-0.1,0)(1.2,0)
\psline{-}(1,0)(1,1.2)
\rput[u](-0.05,-0.05){$0$}
\rput[u](-0.12,1.15){$u^*_S(\nu)$}
\rput[u](-0.05,1){$1$}
\rput[u](1,-0.05){$1$}
\rput[u](1.3,-0.04){$\nu(\omega_1)$}
\rput[u](0.875,-0.09){$\frac78$}
\rput[u](0.5,-0.09){$\frac12$}
\rput[u](-0.3,0.571){$\cav u^*_S(\frac12)=\frac47$}
\psline[linestyle=dotted](0,1)(0.875,1)
\psline[linecolor=blue,linewidth=3.5pt](0.875,1)(1,1)
\psline[linecolor=blue,linewidth=3.5pt](0,0)(0.875,0)
\psline[linestyle=dotted](0.875,0)(0.875,1)
\psline[linestyle=dashed, linecolor=blue](0,0)(0.875,1)
\psline[linestyle=dashed](0.5,0)(0.5,0.571)
\psline[linestyle=dotted](0.5,0.571)(0,0.571)
\psdots[linecolor = red,dotscale=2 2](0,0)(0.875,1)
\psdots[linecolor=blue](0.5,0.571)
\psdots[linecolor=blue,dotscale=2 2](0.875,0)
\end{pspicture}
\caption{Concavification.}
\label{fig:example01}
\end{center}
\end{figure}
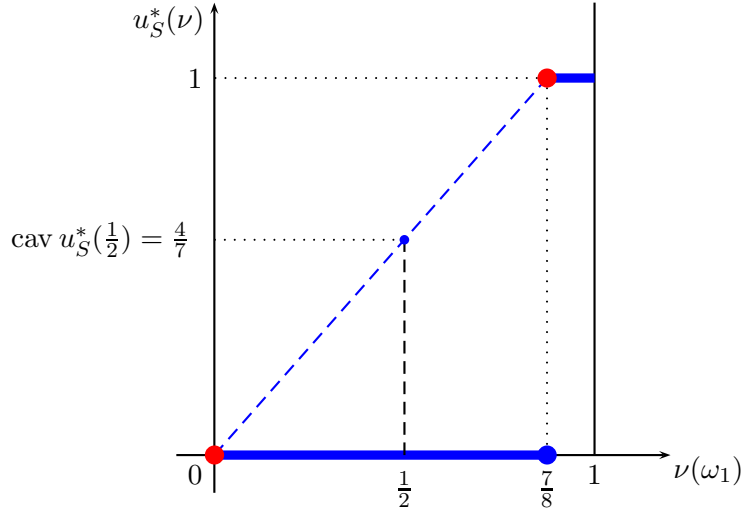

The concavification function $\cav u^*_S(\nu)$ is continuous and equal to $\frac87 \nu(\omega_1)$ for $\nu(\omega_1)\leq \frac78$ and 1 otherwise. It is easy to see that it does not depend on the action chosen by the receiver at $\nu(\omega_1)=\frac78$, see Figure \ref{fig:example01}.
If the receiver were to choose $a_1$ at the point of indifference, then the optimal splitting for the sender would be as follows:
$$\bigg(\frac12,\frac12\bigg)=\frac37\big(1,0\big)+\frac47\bigg(\frac18,\frac78\bigg),$$
where a belief is denoted $\nu=(\nu(\omega_0),\nu(\omega_1))$. This yields a payoff of $\frac47$ which is the highest that the sender can achieve given the uniform prior. For any small $\vare>0$, we can perturb the previous splitting and get the following:
$$\bigg(\frac12,\frac12\bigg)=\frac{3+8\vare}{7+8\vare}\big(1,0\big)+\frac{4}{7+8\vare}\bigg(\frac18-\vare,\frac78+\vare\bigg),$$
which achieves the payoff $\frac{4}{7+8\vare}$ irrespective of the tie-breaking rule. Letting $\vare$ tend 0, we see that the sender achieves a payoff arbitrarily close to $\frac47$, which is the optimal robust payoff.

\subsection{Restricted and noisy communication} 
We consider binary sets of messages  $X=\{x_0,x_1\}$, $Y=\{y_0,y_1\}$ and we assume that the channel has a {\em noise level} $\vare\in[0,\frac12]$, that is $Q(y_j|x_i)=\vare$ for $j\neq i$, see Figure \ref{fig:binary}. The generic case is $\vare\in(0,\frac12)$ where the label of the message ($0$ or $1$) is changed with positive probability but observing a label $1$ is still more likely when the input label is $1$. When $\vare=\frac12$, the distribution of the output message is independent from the input message, so the channel completely disrupts the communication.

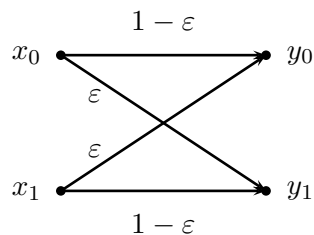
\begin{figure}[!ht]
\begin{center}
\psset{xunit=0.9cm,yunit=0.9cm}
\begin{pspicture}(-4.5,0)(6,3.5)
\rput(-1,0.5){$x_1$}
\rput(-1,2.5){$x_0$}
\psdots(-0.5, 0.5)(-0.5, 2.5)(2.5,0.5)(2.5,2.5)
\psline[linewidth=1pt]{->}(-0.5,0.5)(2.5,0.5)
\psline[linewidth=1pt]{->}(-0.5,2.5)(2.5,2.5)
\psline[linewidth=1pt]{->}(-0.5,0.5)(2.5,2.5)
\psline[linewidth=1pt]{->}(-0.5,2.5)(2.5,0.5)
\rput(3,0.5){$y_1$}
\rput(3,2.5){$y_0$}
\rput(1,3){$1- \vare$}
\rput(1,0){$1- \vare$}
\rput(0,1.9){$ \vare$}
\rput(0,1.1){$\vare$}
\end{pspicture}
\caption{Binary symmetric channel.}
\label{fig:binary}

\end{center}
\end{figure}

A special case is the {\em binary perfect channel} when $\vare=0$: identifying together the sets $X$ and $Y$, an input message $x$ is received with certainty. Communication is then restricted only by the number of available messages, i.e. the cardinality of $X$. 

The capacity of the binary symmetric channel\footnote{\citealp[Example 2.1.1, p. 15]{cover-book-2006}.} is  $1-H(\varepsilon)$ where with some abuse of notation, $H(\varepsilon)$ denotes the entropy of the binary probability distribution $(\varepsilon,1- \varepsilon)$.

\subsection{One-shot scenario $k = n = 1$} 

Let a strategy $\sigma$ of the sender be parametrized by $\sigma(x_0|\omega_0)=1-\alpha$ and $\sigma(x_1|\omega_1)=1-\beta$; see Figure \ref{fig:strategy}. 

\begin{figure}[!ht]
\begin{center}
\psset{xunit=0.75cm,yunit=0.7cm}
\begin{pspicture}(-2,-1)(7,3.3)
\rput(-1,0.5){$\omega_1$}
\rput(-1,2.5){$\omega_0$}
\rput(-3,0.5){$\mu(\omega_1)$}
\rput(-3,2.5){$\mu(\omega_0)$}
\psdots(-3.5,1.5)(-0.5, 0.5)(-0.5, 2.5)(2.5,0.5)(2.5,2.5)
\psline[linewidth=1pt]{->}(-3.5,1.5)(-1.5,2.5)
\psline[linewidth=1pt]{->}(-3.5,1.5)(-1.5,0.5)
\psline[linewidth=1pt]{->}(-0.5,0.5)(2.5,0.5)
\psline[linewidth=1pt]{->}(-0.5,2.5)(2.5,2.5)
\psline[linewidth=1pt]{->}(-0.5,0.5)(2.5,2.5)
\psline[linewidth=1pt]{->}(-0.5,2.5)(2.5,0.5)
\rput(3,0.5){$x_1$}
\rput(3,2.5){$x_0$}
\psdots(3.5, 0.5)(3.5, 2.5)(6.5,0.5)(6.5,2.5)
\rput(7,0.5){$y_1$}
\rput(7,2.5){$y_0$}
\psline[linewidth=1pt]{->}(3.5,0.5)(6.5,0.5)
\psline[linewidth=1pt]{->}(3.5,2.5)(6.5,2.5)
\psline[linewidth=1pt]{->}(3.5,0.5)(6.5,2.5)
\psline[linewidth=1pt]{->}(3.5,2.5)(6.5,0.5)
\rput(5,3){$1 - \varepsilon$}
\rput(5,0){$1 - \varepsilon$}
\rput(3.8,1.9){$ \varepsilon$}
\rput(3.8,1.1){$ \varepsilon$}
\rput(1,3){$1 - \alpha$}
\rput(1,0){$1 - \beta$}
\rput(-0.2,1.9){$ \alpha$}
\rput(-0.2,1.1){$ \beta$}
\end{pspicture}
\caption{Strategy on the binary symmetric channel.}
\label{fig:strategy}
\end{center}
\end{figure}
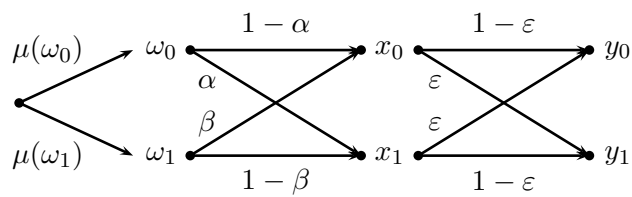

Then,
$\prob_\sigma(y_1|\omega_0)=\alpha(1-\vare)+(1-\alpha)\vare,$
$\prob_\sigma(y_0|\omega_1)=\beta(1-\vare)+(1-\beta)\vare$ and from Bayes' rule,
$$\prob_\sigma(\omega_1|y_1)=\frac{\mu(\omega_1)(1-\prob_\sigma(y_0|\omega_1))}{\mu(\omega_0)\prob_\sigma(y_1|\omega_0)+\mu(\omega_1)(1-\prob_\sigma(y_0|\omega_1))},$$
$$\prob_\sigma(\omega_1|y_0)=\frac{\mu(\omega_1)\prob_\sigma(y_0|\omega_1)}{\mu(\omega_0)(1-\prob_\sigma(y_1|\omega_0))+\mu(\omega_1)\prob_\sigma(y_0|\omega_1)}.$$ 
It is easy to see that the numbers $\prob_\sigma(y_1|\omega_0)$, $\prob_\sigma(y_0|\omega_1)$, 
$\prob_\sigma(\omega_1|y_1)$, $\prob_\sigma(\omega_1|y_0)$ all belong to the interval $[\vare,1-\vare]$. 

A pair of posteriors $(\nu_0,\nu_1)$ is said to be {\em feasible} in the one-shot scenario if there exists a number $\lambda\in[0,1]$ such that: $$(\mu(\omega_0),\mu(\omega_1))=\lambda(\nu_0(\omega_0),\nu_0(\omega_1))+(1-\lambda)(\nu_1(\omega_0),\nu_1(\omega_1)).$$
The feasible splittings can be characterized as follows.
\begin{lemma}\label{lem:feasible}
We consider the one-shot problem where $n=k=1$. A pair of posteriors $(\nu_0,\nu_1)$ is feasible if and only if $\nu_1=\nu_0=\mu$ or,
$$\vare\leq \frac{\nu_0(\omega_1)(\nu_1(\omega_1)-\mu(\omega_1))}{\mu(\omega_1)(\nu_1(\omega_1)-\nu_0(\omega_1))}\leq 1-\vare$$
and
$$\vare\leq \frac{(1-\nu_0(\omega_1))(\mu(\omega_1)-\nu_0(\omega_1))}{(1-\mu(\omega_1))(\nu_1(\omega_1)-\nu_0(\omega_1))}\leq 1-\vare.$$
\end{lemma}
The proof is in Appendix \ref{A11}. As an illustration, take the uniform prior $(\frac12,\frac12)$ and a level of noise $\vare=\frac14$. The feasible posteriors are shown by the colored green regions on Figure \ref{fig:PosteriorsAsymptotic0}.

From the previous discussion, it is impossible to induce beliefs with $\nu(\omega_1)>\frac34$. 
Therefore, the receiver will never be confident enough to invest and the payoff is 0 for the sender.

\subsection{Asymptotic scenario with $k = n \to \infty$}\label{mainex2}
We consider the case where $k=n$ tends to infinity with a noise level of $\vare=\frac14$ and compute the value of the  optimal splitting problem with information constraint. The capacity of the channel is $1-H(\frac14)$, the entropy of the uniform prior is 1; therefore, the information constraint is $\sum_m\lambda_mH(\mu_m)\geq H(\frac14)$. 
Under this constraint the optimal splitting for the sender satisfies:
$$\bigg(\frac12,\frac12\bigg)=\lambda\bigg(\frac18,\frac78\bigg)+(1-\lambda)(\nu_0(\omega_0),\nu_0(\omega_1))$$
and
$$H\bigg(\frac14\bigg)=\lambda H\bigg(\frac78\bigg)+(1-\lambda)H(\nu_0(\omega_1)).$$
To see why it is optimal, first observe that the sender has to bring on some  posterior, denoted by $\nu_1$, with $\nu_1(\omega_1)>\frac78$ in order to get some payoff. To get it with the highest probability, he should aim for $\nu_1(\omega_1)=\frac78$. Among the posteriors that induce investment, this is also the one with highest entropy. Second, to maximize expected payoffs, the remaining posteriors must be as far away as possible from the prior; that is, the information constraint should bind. Additionally, note that only one posterior, denoted by $\nu_0$, will be optimally generated in the region $\nu_0(\omega_1)<\frac78$. Since the entropy is strictly concave, replacing two posteriors on this region by their average does not change the payoff and increases the entropy.

Solving these two equations numerically we get, $\nu_0(\omega_1)\approx 0.340$ and $V(\mu,Q)=\lambda\approx0.298$ instead of the zero value for the  one-shot scenario and about $52.1\%$ of the unconstrained optimum $\frac47$.

This is shown in Figure \ref{fig:optimalentropy01} which plots the payoff function and the entropy function. The splitting of $\mu$ into $\nu_0,\nu_1$ is shown by the three points on the horizontal axis. On the vertical line $\mu(\omega_1)=\frac12$,  we can read the average payoff with the red line and the average entropy with the green line. To see optimality on the picture, if we move $\nu_0(\omega_1)$ to the right, then the average payoff decrease, and if we move it to the left, the average entropy will fall below $H(\frac14)$ and the information constraint will be violated.

\begin{figure}[!ht]
\begin{center}
\psset{xunit=5cm,yunit=5cm}
\begin{pspicture}(0,-0.4)(1,1.3)
\psplot[plotpoints=100]{0.001}{0.999}{x ln 2 ln div x mul neg 1 x neg add ln 2 ln div 1 x neg add mul neg add}
\psline{->}(0,-0.1)(0,1.2)
\psline{->}(-0.1,0)(1.2,0)
\psline{->}(1,-0.1)(1,1.2)
\rput[u](-0.03,-0.04){$0$}
\rput[u](-0.03,1){$1$}
\rput[u](1.03,1){$1$}
\rput[u](1.03,-0.04){$1$}
\rput[u](1.3,-0.03){$\nu(\omega_1)$}
\rput[u](1.12,1.15){$H(\nu)$}
\rput[u](-0.1,1.15){$u_S^*(\nu)$}
\psline[linestyle=dotted](0,1)(1,1)
\rput[bl]{-45}(0.5,-0.08){$\mu(\omega_1) = \frac12$}
\psline[linecolor=blue,linewidth=3.5pt](0,0)(0.875,0)
\psline[linecolor=blue,linewidth=3.5pt](0.875,1)(1,1)
\psline[linestyle=dotted](0.875,0)(0.875,1)
\psline[linestyle=dotted](0.3404,0)(0.3404,0.9252)
\psline[linestyle=dashed](0.5,0)(0.5,1)
\psline[linestyle=dashed](0,0.81)(1,0.81)
\rput[u](1.12,0.81){$H(\frac14)$}
 \psline[linestyle=dotted, linecolor=black](0.5,0.298)(0,0.298)
\rput[r](-0.015,0.298){$V(\mu,C(Q))\approx0.298$}
\psdots[linecolor = red,dotscale=2 2](0.3404,0)(0.875,1)
\psline[linestyle=solid, linecolor=red](0.34,0)(0.875,1)
\psdots[linecolor = red](0.5,0.298)
\psdots[linecolor = green](0.875,0.5436)(0.3404,0.9252)(0.5,0.81)
\psline[linestyle=solid, linecolor=green](0.875,0.5436)(0.3404,0.9252)
\rput[bl]{-45}(0.845,-0.08){$\nu_1(\omega_1) = \frac78$}
\rput[bl]{-45}(0.32,-0.08){$\nu_0(\omega_1) \approx 0.340$}
\end{pspicture}
\caption{For a noise parameter $\vare=\frac14$, the optimal splitting is given by $\nu_0(\omega_1)\approx 0.340$ and $\nu_1(\omega_1) = \frac78$.}
 \label{fig:optimalentropy01}
\end{center}
\end{figure}
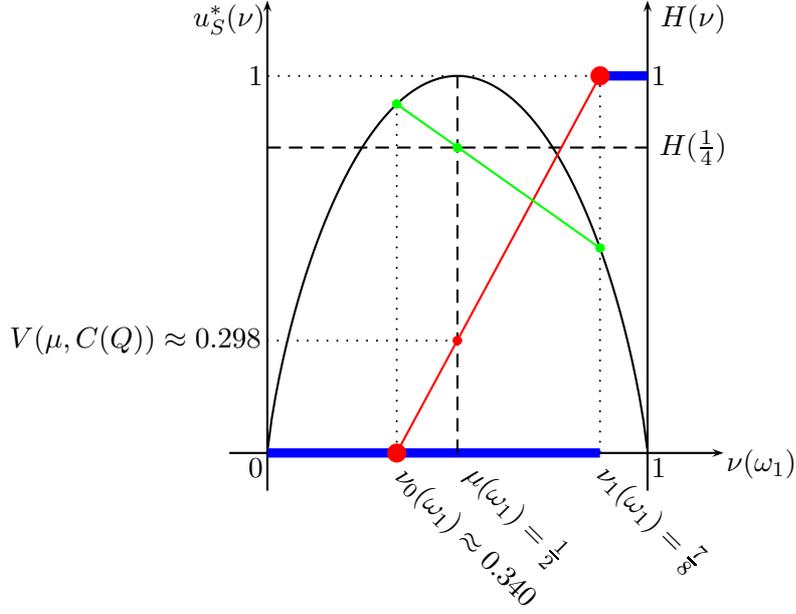

The optimal splitting is also marked on Figure \ref{fig:PosteriorsAsymptotic0} which shows the set of pairs of posteriors for the splittings that satisfy the information constraint (union of green and blue regions).

\begin{figure}[!ht]\label{fig:feasible}
\begin{center}
\psset{xunit=6cm,yunit=6cm}
\begin{pspicture}(0,-0.3)(1,1.2)
\pscustom[linewidth=0pt,linestyle=none,fillstyle=solid,fillcolor=blue]{
\fileplot[]{DataIC/Data1IC0.25.dat}
\psline{}(0.5,1)(0.5,0.5)(0,0.5)
}
\fileplot[]{DataIC/Data1IC0.25.dat}
\pscustom[linewidth=0pt,linestyle=none,fillstyle=solid,fillcolor=blue]{
\fileplot[]{DataIC/Data2IC0.25.dat}
\psline{}(1,0.5)(0.5,0.5)(0.5,0)
}
\fileplot[]{DataIC/Data2IC0.25.dat}
\pspolygon[fillstyle=vlines](0,0)(0.5,0)(0.5,0.5)(0,0.5)
\pspolygon[fillstyle=vlines](1,1)(0.5,1)(0.5,0.5)(1,0.5)
\pscustom[linewidth=0pt,linestyle=none,fillstyle=vlines,]{
\fileplot[]{DataIC/Data2IC0.25.dat}
\psline{}(1,0)
}
\pscustom[linewidth=0pt,linestyle=none,fillstyle=vlines,]{
\fileplot[]{DataIC/Data1IC0.25.dat}
\psline{}(0,1)
}
\pscustom[linewidth=0pt,linestyle=none,fillstyle=solid,fillcolor=green]{
\psplot[plotpoints=100,linecolor=brown]{0.25}{0.5}{0.875 x mul  0.5 neg add 0.625 neg x add div}
\psplot[plotpoints=100,linecolor=blue]{0.5}{0.25}{0.375 x mul  0.125 neg x add div}
}
\pscustom[linewidth=0pt,linestyle=none,fillstyle=solid,fillcolor=green]{
\psplot[plotpoints=100,linecolor=red]{0.5}{0.75}{0.125 x mul  0.375 neg x add div}
\psplot[plotpoints=100,linecolor=green]{0.75}{0.5}{0.625 x mul  0.5 neg add 0.875 neg x add div}
}
\psline{->}(0,-0.1)(0,1.2)
\psline{->}(-0.1,0)(1.2,0)
\psline{-}(1,-0.1)(1,1.2)
\psline{-}(-0.1,1)(1.2,1)
\psline{-}(0.5,0)(0.5,1)
\psline{-}(0,0.5)(1,0.5)
\psline[linestyle=dotted](0.25,0)(0.25,1)
\psline[linestyle=dotted](0.75,0)(0.75,1)
\psline[linestyle=dotted](0,0.25)(1,0.25)
\psline[linestyle=dotted](0,0.75)(1,0.75)
\rput[u](-0.03,-0.03){$0$}
\rput[u](-0.03,1.03){$1$}
\rput[u](-0.05,0.5){$\frac12$}
\rput[u](-0.05,0.25){$\frac14$}
\rput[u](-0.05,0.75){$\frac34$}
\rput[u](1.03,-0.03){$1$}
\rput[u](0.5,-0.06){$\frac12$}
\rput[u](0.25,-0.06){$\frac14$}
\rput[u](0.75,-0.06){$\frac34$}
\rput[u](-0.1,1.15){$\nu_0(\omega_1)$}
\rput[u](1.25,-0.03){$\nu_1(\omega_1)$}
\psdots[linecolor = red,dotscale=2 2](0.875,0.3404)
\psline[linestyle=solid, linecolor=red](0,0.3404)(0.875,0.3404)(0.875,0)
\rput[r](-0.02,0.3404){\textcolor[rgb]{1,0,0}{$\nu_0(\omega_1)  \approx 0.340$}}
\rput[bl]{-45}(0.845,-0.06){\textcolor[rgb]{1,0,0}{$\nu_1(\omega_1) = \frac78$}}
\end{pspicture}
\caption{For a noise parameter $\vare=\frac14$, the green lenses correspond to the feasible posteriors $(\nu_0,\nu_1)$ characterized in Lemma \ref{lem:feasible} for the one-shot scenario $k=n=1$. The blue and green regions correspond to the feasible posteriors $(\nu_0,\nu_1)$ in the asymptotic scenario where $k = n \to \infty$. The red point corresponds to the optimal splitting, also depicted in Figure \ref{fig:optimalentropy01}. The hatched areas correspond to the nonfeasible posteriors $(\nu_0,\nu_1)$. }
\label{fig:PosteriorsAsymptotic0}
\end{center}
\end{figure}
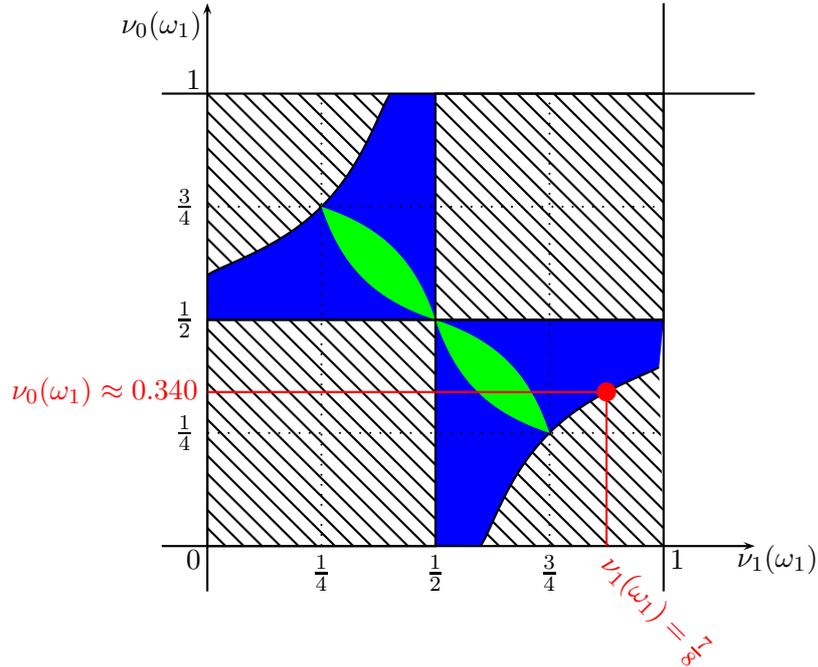

\newpage

On Figure \ref{fig:ValueNoise}, we represent the value $V(\mu,C(Q))$ of the optimal splitting problem as a function of the prior $\mu$, for different values for the noise parameter $\vare\in \Big\{\frac{1}{20},\frac{3}{20}, \frac{1}{4},\frac{7}{20},\frac{9}{20},\frac{99}{200}\Big\}$. It is found by solving the following system for $\nu_0$:
$$(\mu(\omega_0),\mu(\omega_1)) =\lambda\bigg(\frac18,\frac78\bigg)+(1-\lambda)(\nu_0(\omega_0),\nu_0(\omega_1))$$
and
$$
H(\mu(\omega_1)) - 1 + H(\varepsilon)=\lambda H\bigg(\frac78\bigg)+(1-\lambda)H(\nu_0(\omega_1)).$$
When $\mu(\omega_1) = \frac12$ and $\varepsilon =  \frac14$, we recover the value $V(\mu,C(Q)) \approx0.298$ as in Figure \ref{fig:optimalentropy01}.

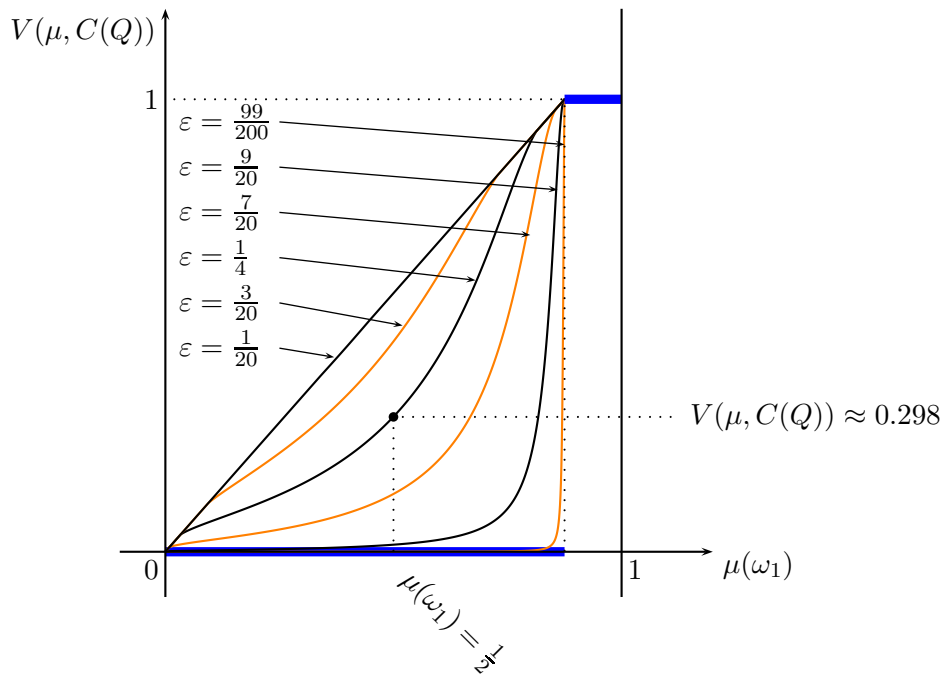
\begin{figure}[!ht]\label{fig:feasible}
\begin{center}
\psset{xunit=6cm,yunit=6cm}
\begin{pspicture}(0,-0.3)(1,1.2)
\rput[u](-0.03,-0.04){$0$}
\rput[u](-0.03,1){$1$}
\rput[u](1.03,-0.04){$1$}
\rput[u](1.3,-0.03){$\mu(\omega_1)$}
\rput[u](-0.18,1.15){$V(\mu,C(Q))$}

\psline[linecolor=blue,linewidth=3.5pt](0,0)(0.875,0)
\fileplot[linecolor=orange]{DataV/DataV_0.495.dat}
\fileplot[linecolor=black]{DataV/DataV_0.45.dat}
\fileplot[linecolor=orange]{DataV/DataV_0.35.dat}
\fileplot[linecolor=black]{DataV/DataV_0.25.dat}
\fileplot[linecolor=orange]{DataV/DataV_0.15.dat}
\fileplot[linecolor=black]{DataV/DataV_0.05.dat}
\psline[linestyle=dotted](0.875,0)(0.875,1)(0,1)
\psline[linecolor=blue,linewidth=3.5pt](0.875,1)(1,1)
\psline{->}(0,-0.1)(0,1.2)
\psline{->}(-0.1,0)(1.2,0)
\psline{-}(1,-0.1)(1,1.2)
\psdot(0.5,0.298)
\rput[l](0.03,0.95){$\varepsilon = \frac{99}{200}$}
\rput[l](0.03,0.85){$\varepsilon = \frac{9}{20}$}
\rput[l](0.03,0.75){$\varepsilon = \frac{7}{20}$}
\rput[l](0.03,0.65){$\varepsilon = \frac14$}
\rput[l](0.03,0.55){$\varepsilon = \frac{3}{20}$}
\rput[l](0.03,0.45){$\varepsilon= \frac{1}{20}$}
\psline[linewidth=0.5pt]{->}(0.25,0.95)(0.875,0.9)
\psline[linewidth=0.5pt]{->}(0.25,0.85)(0.86,0.8)
\psline[linewidth=0.5pt]{->}(0.25,0.75)(0.8,0.7)
\psline[linewidth=0.5pt]{->}(0.25,0.65)(0.68,0.6)
\psline[linewidth=0.5pt]{->}(0.25,0.55)(0.52,0.5)
\psline[linewidth=0.5pt]{->}(0.25,0.45)(0.37,0.43)

\psline[linestyle=dotted, linecolor=black](0.5,0)(0.5,0.298)(1.13,0.298)
\rput[l](1.15,0.298){$V(\mu,C(Q)) \approx0.298$}
\rput[bl]{-45}(0.5,-0.08){$\mu(\omega_1) = \frac12$}
\end{pspicture}
\caption{Value of the optimal splitting problem as a function of the prior $\mu$, for different noise parameters $\varepsilon \in \Big\{\frac{1}{20},\frac{3}{20}, \frac{1}{4},\frac{7}{20},\frac{9}{20},\frac{99}{200}\Big\}$.}
\label{fig:ValueNoise}
\end{center}
\end{figure}

Observe that the function $V(\mu,C(Q))$ is {\em not} concave with respect to the prior $\mu$. From Corollary \ref{coro:Cav}, $V(\mu, c)$ is the concavification of the function $u^H_S$ calculated at $(\mu, H(\mu)-c)$, so this composed function need not be concave.

\subsection{Perfect binary channel with $k = \frac{n}{2} \to \infty$}
We consider the same example as before, repeated $n$ times with the uniform prior $\mu=(\frac12,\frac12)$. In line with the motivating example from the introduction, we consider a perfect channel and assume that the sender has at its disposal half as many messages as needed to communicate perfectly, that is $k = \frac{n}{2}$. Since the capacity of the binary perfect channel is one, $\frac{k}{n}=\frac 12$ and $H(\mu)=1$, the information constraint is:
$$H(\mu) - \sum_m\lambda_m H(\mu_m) \leq \frac12\Longleftrightarrow \sum_m\lambda_m H(\mu_m) \geq \frac12.$$
Observe that this constraint is identical to the one obtained with a binary symmetric channel with noise $\varepsilon$ such that $H(\varepsilon) = \frac12$ (i.e., $\vare\approx 0.110$).
Therefore, the optimal splitting is given by the following system:
$$(\mu(\omega_0),\mu(\omega_1)) =\lambda\bigg(\frac18,\frac78\bigg)+(1-\lambda)(\nu_0(\omega_0),\nu_0(\omega_1))$$
and
$$\frac12=\lambda H\bigg(\frac78\bigg)+(1-\lambda)H(\nu_0(\omega_1)).$$
Solving numerically, we find $V(\mu,\frac12) \approx 0.519$.

\subsection{On the number of posteriors}
We give now an example showing the tightness of the bound $\min\{|A|, |\Omega|+1\}$ on the number of posteriors, given in Corollary \ref{lem:actions}. The payoff table is as follows:

\begin{center}
\begin{tabular}{ccccc}
                            & $a_0$ & $a_1$& $a_2$& \\ 
 \cline{2-4}
\multicolumn{1}{l|}{$\omega_0$} & \multicolumn{1}{l|}{$\,\,0,0\,\,$} & \multicolumn{1}{l|}{$1,-7$} &  \multicolumn{1}{l|}{$\,\,1,1$}& $\frac12$  \\ \cline{2-4}
\multicolumn{1}{l|}{$\omega_1$}   & \multicolumn{1}{l|}{$\,\,0,0\,\,$} & \multicolumn{1}{l|}{$\,\,1,1$} & \multicolumn{1}{l|}{$1,-7$} & $\frac12$  \\ \cline{2-4}
&&&&
\end{tabular}
\end{center}

\vskip0.5cm
\noindent There are two risky projects ($a_1$ and $a_2$) and the sender wants to persuade the receiver to invest in any of them. The receiver invests only if $\nu(\omega_1)> 7/8$ or $\nu(\omega_1)<1/8$.

With unrestricted communication, the solution is clear: the sender fully discloses the state and gets a payoff of 1. However, with a binary symmetric channel with noise $\vare=1/4$, the sender gets 0 in the one-shot scenario. Consider now the case where $k=n\to \infty$.

The ``one-sided'' solution of Section \ref{mainex2} is feasible. Recall that this is the splitting such that:
$$\bigg(\frac12,\frac12\bigg)=\lambda\bigg(\frac18,\frac78\bigg)+(1-\lambda)(\nu_0(\omega_0),\nu(\omega_1))$$
and
$$H\bigg(\frac14\bigg)=\lambda H\bigg(\frac78\bigg)+(1-\lambda)H(\nu_0(\omega_0),\nu_0(\omega_1)).$$
with  $\nu_0(\omega_1)\approx 0.340$ and $\lambda\approx0.298$. It is easy to see that this is optimal among the splittings with two posteriors. Indeed, it is not possible that the two posteriors induce investment while satisfying the information constraint.

However, this is not optimal. The optimal splitting has three posteriors and is as follows:
$$\bigg(\frac12,\frac12\bigg)=(1-\lambda)\bigg(\frac12,\frac12\bigg)+\frac{\lambda}{2}\bigg(\frac18,\frac78\bigg)+ \frac{\lambda}{2}\bigg(\frac78,\frac18\bigg)$$
with
$$H\bigg(\frac14\bigg)=(1-\lambda)H\bigg(\frac12\bigg)+\frac{\lambda}{2} \, H\bigg(\frac18\bigg)+ \frac{\lambda}{2}\,  H\bigg(\frac78\bigg).$$
This pins down a unique $\lambda$ and solving numerically yields $\lambda\approx 0.413.$ Since $\lambda$ is the probability of investment, we get $V(\mu,Q)\approx 0.413$ which is about $38\%$ better than what is achieved with a splitting with two points.

To see that this is optimal, first since there are two states, we know that three posteriors are sufficient. Second, it is not possible to have all posteriors in the investment region and to satisfy the information constraint. If there is only one posterior in the investment region, then the splitting achieves no more than the ``one-sided'' solution. Therefore, it is optimal to have two posteriors in the investment region and one outside of it. However, then, it is optimal to choose the point in the middle region to be $(\frac12,\frac12)$, since this is the one with the highest entropy.

\section{Beyond identical problems}\label{sec:extension}
The main result can  be extended to series of persuasion problems which are not all identical, but such that each type of problem is repeated many times.
Suppose that we have a family of persuasion problems indexed by a type parameter $z$ in a finite set $Z$. That is, for every $z\in Z$, there is a prior probability distribution $\mu(\cdot |z) \in \Delta(\Omega)$ and  payoff functions $u_i(\cdot,z):\Omega\times A\to \reals$ for each player $i=S,R$. The series of persuasion problems is given by a sequence $z^n=(z_1,\dots,z_n)$ which is commonly known by both players. The distribution of states is as follows:
$$\mu^n(\omega^n|z^n):=\prod_{t=1}^n\mu(\omega_t|z_t).$$
If the sequence of states and actions are respectively $\omega^n, a^n$, the payoff for player $i$ is $\frac1n\sum_{t=1}^n u_i(\omega_t,a_t,z_t)$.
The communication technology is still given by a channel $Q:X\to\Delta(Y)$ used $k$ times, so that the strategy sets are the same as before for both players. The optimal robust payoff of the sender is defined as before and is denoted by $U^*_S(\mu^n,Q^k,z^n)$.

For each posterior belief $\nu\in\Delta(\Omega)$ and type $z\in Z$, the set of optimal actions of the receiver is
$A^{*z}(\nu)=\argmax_{a\in A}\sum_\omega \nu(\omega)u_{R}(\omega,a,z)$ and 
we denote  by $u^{*z}_S(\nu)=\min_{a\in A^{*z}(\nu)}\sum_\omega \nu(\omega)u_{S}(\omega,a,z)$ the robust payoff of the sender at the belief $\nu$.

\begin{definition} For $\pi\in\Delta(Z)$ and $c\ge0$, the optimal splitting problem with information constraint is as follows: \begin{eqnarray*}
V^Z(\mu,c,\pi)=&\sup & \sum_{z} \pi(z)  \sum_{m}  \lambda^z_m u^{*z}_S(\nu^z_m)\\
&\mathrm{ s.t.} &  \sum_m\lambda^z_m\nu^z_m=\mu(\cdot | z), \quad \forall z\in Z,\\
&\mathrm{ and } &  \sum_z \pi(z) \Big(H(\mu(\cdot | z))-\sum_m\lambda^z_m H(\nu^z_m)\Big) \leq c.
\end{eqnarray*}
\end{definition} 
The interpretation is as follows. Suppose that $\pi(z)$ represents the probability, or frequency, of occurrence of $z$. Conditional on $z$ which is known by both players, the sender performs a spitting of
$\mu(\cdot | z)$, $\sum_m\lambda^z_m\nu^z_m=\mu(\cdot |z)$,
and gets the payoff $\sum_{m}  \lambda^z_m u^{*z}_S(\nu^z_m)$. The information constraint imposes the {\em average} mutual information to be less than or equal to the capacity.

Given a sequence $z^n\in Z^n$, let  $\pi_{n}\in\Delta(Z)$ be  the empirical frequency induced by the sequence: for each $z\in Z$, $\pi_{n}(z)=\frac{1}{n}|\{t : z_t=z\}|$.

\begin{theorem}\label{mainthmZ}
\begin{enumerate}
\item The optimal robust payoff of the sender is no more than the value of the optimal splitting problem with information constraint. For each pair of integers $n,k$:

$$U^*_S(\mu^{n},Q^k,z^n) \leq V^Z(\mu,\frac{k}{n}C(Q),\pi_{n}).$$

\item  The optimal robust payoff of the sender converges to the value of the optimal splitting problem with information constraint in the following sense. For each $\pi\in\Delta(Z)$ and $r\in[0,+\infty]$, for each pair of sequences of integers $(k_j,n_j)_{j\in\naturals}$ such that $\lim\limits_{j \to \infty}\max(n_j,k_j)= \infty$, $\lim\limits_{j \to \infty}\frac{k_j}{n_j}=r$ and $\lim\limits_{j \to \infty}\pi_{n_j}=\pi$, we have:

$$\lim_{j\to \infty}U^*_S(\mu^{n_j},Q^{k_j},z^{n_j}) = V^Z(\mu,rC(Q),\pi).$$

\end{enumerate}
\end{theorem}
Given a sequence  $z^n$, $\pi_{n}$ is the empirical distribution of types of problems. The optimal payoff of the sender is bounded above by the value of optimal splitting under information constraint. Suppose that the distribution of types is held fixed (or converges to) $\pi$, then when $n$ and $k$ grow large, the sender is able to secure approximately this value. The arguments of the proof of Theorem \ref{mainthm} extend quite easily to this case (up to lengthy adaptations for the second point) so the proof is omitted.

This extension applies to the case where the proportions of types of problems are fixed. Alternatively, the sequence $z^n$ could be drawn i.i.d. from a prior distribution $\pi\in\Delta(Z)$.

Notice the channel $Q^k$ is used for transferring information about all problems. Thus, Theorem \ref{mainthmZ} does more than merely patching up distinct families of problems together. The capacity of the channel bounds the total amount of information, across all problems. Thus, all problems, even of different types, are  linked together in the messages.

\section{Conclusion}\label{sec:conclusion}
We have analyzed a persuasion game where the sender communicates with the receiver through a fixed and imperfect channel. The optimal payoff of the sender is bounded above by the value of the optimal splitting problem with information constraint. When the sender and the receiver are engaged in many repetitions of identical persuasion games, the optimal payoff for the sender converges to the upper bound as the number of repetitions increases.

There are several interesting variations or extensions of this model. 

\begin{enumerate}

\item {\em Private information of the receiver}. In the model, it is assumed that the information about the state if fully controlled by the sender. To model private information of the receiver, consider the extension of the previous section, let nature draw  pairs $(\omega_t, z_t)_{t=1,\dots, n}$  and assume that
$\omega^n$ is the private information of the sender and $z^n$ is the private information of the receiver. Our methods generalize to this case provided that we use a suitable generalization of the information constraint. A random message $\rv m$ can be transmitted over the channel provided that its mutual information with the state, {\em conditional on the private information of the receiver} $I(\rv \omega;\rv m |\rv z)\leq C$ is less than or equal to the capacity,  where  
$$I(\rv \omega;\rv m |\rv z):=\sum_z\prob(z)I(\rv \omega;\rv m | \rv z=z)$$
is the expectation over $\rv z$ of the mutual information conditional on $\{\rv z=z\}$.

\item {\em Commitment of the receiver}. In the persuasion model, the sender first chooses  its strategy and is committed to playing it.  A natural twist is to let the receiver choose his strategy first and commit to it. This turns into a mechanism design problem where the receiver is a principal offering a contract to an informed agent (the sender), and where the agent communicates with the principal through an imperfect channel. Again, an information constraint holds, but the impact of incentives is different. Namely for each sequence of states $(\omega_1,\dots, \omega_n)$, the sender/agent should not have an incentive to behave as if it was another one $(\omega'_1,\dots, \omega'_n)$. The task is to prove that using the usual coding scheme is indeed an optimal strategy for the agent, or rather that any optimal strategy is not too different from the coding scheme. This variation is  under close study.

\item {\em More general processes}. It would be interesting to generalize the results to a larger class of stochastic processes of states. Information theoretic methods can be extended to Markov chains, see \cite{cover-book-2006} and to more general processes, see \cite{Han03}. While an information constraint would certainly hold, it is an open problem to characterize the optimal payoff for the sender. What is the best way to exploit the correlations between states?
\end{enumerate}

\bibliographystyle{econometrica}

\ifx\undefined\BySame
\newcommand{\BySame}{\leavevmode\rule[.5ex]{3em}{.5pt}\ }
\fi
\ifx\undefined\textsc
\newcommand{\textsc}[1]{{\sc #1}}
\newcommand{\emph}[1]{{\em #1\/}}
\let\tmpsmall\small
\renewcommand{\small}{\tmpsmall\sc}
\fi



\appendix

\section{ Appendix}\label{appendix}
This appendix contains all the formal proofs. 

\subsection{Proof of Lemma \ref{lem:feasible}}\label{A11}

For $a,b$ in $[0,1]$, consider the system:
\begin{equation}\label{eq:feasible}\nu_1(\omega_1)=\frac{\mu(\omega_1)(1-b)}{\mu(\omega_0)a+\mu(\omega_1)(1-b)}, \hskip0.5cm\nu_0(\omega_1)=\frac{\mu(\omega_1)b}{\mu(\omega_0)(1-a)+\mu(\omega_1)b}.
\end{equation}
If $\nu_1=\nu_0=\mu$, then it must be that $a=1-b$. Otherwise, $\nu_1(\omega_1)\neq\nu_0(\omega_1)$. It is easily verified that the system has a unique solution given by:
$$b= \frac{\nu_0(\omega_1)(\nu_1(\omega_1)-\mu(\omega_1))}{\mu(\omega_1)(\nu_1(\omega_1)-\nu_0(\omega_1))}$$
and
$$a= \frac{(1-\nu_0(\omega_1))(\mu(\omega_1)-\nu_0(\omega_1))}{(1-\mu(\omega_1))(\nu_1(\omega_1)-\nu_0(\omega_1))}.$$
Take a strategy $\sigma$ defined by $\sigma(x_0|\omega_0)=1-\alpha$ and $\sigma(x_1|\omega_1)=1-\beta$ and a  binary symmetric channel with noise $\vare$. The posteriors $\nu_1,\nu_0$ are given by the system (\ref{eq:feasible}) for $a:=\alpha(1-\vare)+\vare(1-\alpha)$ and $b:=\beta(1-\vare)+\vare(1-\beta)$.
 As $\alpha,\beta$ vary in $[0,1]$, $a$ and $b$ range freely over $[\vare,1-\vare]$,
$$\{(\alpha(1-\vare)+\vare(1-\alpha), \beta(1-\vare)+\vare(1-\beta)) : (\alpha,\beta)\in[0,1]^2\}=[\vare,1-\vare]^2.$$
This concludes the proof. \qed

\subsection{Proofs for Sections \ref{sec:cav}}\label{A21}
\subsubsection{Proof of  Theorem \ref{thm:cavcontraint}, point 1}

The function $\cav f^g(x,\gamma)$ is given by the following program:
\begin{eqnarray*}
&  \sup & \sum_m\lambda_mf(x_m)\\
 &  \mathrm{ s.t.} & \sum_m\lambda_mx_m=x, \sum_m\lambda_m\gamma_m=\gamma\\
 & \mathrm{ and\,} &  \forall m, \gamma_m\leq g(x_m).
\end{eqnarray*}
Take a family $(\lambda_m,x_m,\gamma_m)_m$  feasible for this program. We have  $ \sum_m\lambda_mg(x_m)\geq \gamma$, thus this family is feasible for $F^g(x,\gamma)$. Therefore, $\cav f^g(x,\gamma)\leq F^g(x,\gamma)$.

Conversely, take a family $(\lambda_m, x_m)_m$ such that $\sum_m\lambda_m x_m=x$ and $\sum_m\lambda_mg(x_m)\geq \gamma$. Let $\bar\gamma=\sum_m\lambda_mg(x_m)$ and for each $m$, $\gamma_m=g(x_m)+\gamma-\bar\gamma$. Then, $\sum_m\lambda_m\gamma_m=\gamma$ and since $\bar\gamma\geq\gamma$, for each $m$, $\gamma_m\leq g(x_m)$. Thus, $(\lambda_m, x_m, \gamma_m)_m$ is feasible for  $\cav f^g(x,\gamma)$ and  $\cav f^g(x,\gamma)\geq F^g(x,\gamma)$. \qed

\subsubsection{Proof of  Theorem \ref{thm:cavcontraint}, point 2}
Recall that the Fenchel conjugate of $f:X\subseteq\reals^d\to\reals$ is 
$f^*(p)=\sup_x\{x\cdot p -f(x)\}$, where $x\cdot p$ denotes the inner product. Then, the largest convex function below $f$ is equal to $(f^{*})^*$ \cite[Corollary 12.1.1, p. 103]{rockafellar1970convex},
therefore $(f^{*})^*(x)=-\cav (-f)(x)$. Playing with signs, it follows that:
\begin{align}
\cav f(x)=\inf_{p}\bigg\{x\cdot p+\sup_y\{f(y)-p\cdot y\}\bigg\}.
\end{align}
We apply this formula to the function:
\[ f^g(x,\gamma)=
\begin{cases}
f(x) & \text{ if } \gamma\leq g(x), \\
-\infty & \text{ otherwise. } \\
\end{cases}
\]
This gives,
\begin{align*}
\cav f^g(x,\gamma)&=\inf_{p,z}\bigg\{ p\cdot x+z\gamma+\sup_{y,\eta}\{f^g(y,\eta)-p\cdot y-z\eta\}\bigg\}\\
&=\inf_{p,z}\bigg\{ p\cdot x+z\gamma+\sup_{y,\eta :\; \eta\leq g(y)}\{f(y)-p\cdot y-z\eta\}\bigg\}.
\end{align*}
If $z>0$ then by letting $\eta\to-\infty$, the $\sup$ is $+\infty$. Therefore, in the infimum we can restrict to $z\leq 0$. Setting $t=-z\geq0$ we get:

\begin{align*}
\cav f^g(x,\gamma)&=\inf_{t\ge0,p}\bigg\{ p\cdot x-t\gamma+\sup_{y,\eta :\; \eta\leq g(y)}\{f(y)-p\cdot y+t\eta\}\bigg\}\\
&=\inf_{t\ge0,p}\bigg\{ p\cdot x-t\gamma+\sup_{y}\{f(y)-p\cdot y+tg(y)\}\bigg\}\\
&=\inf_{t\ge0}\bigg\{\inf_p\Big\{p\cdot x+\sup_{y}\{f(y)+tg(y)-p\cdot y\}\Big\}-t\gamma\bigg\}
\end{align*}
where the second line holds since $t\geq 0$ and the third line is just reorganizing. The result follows by remarking that $\inf_p\Big\{p\cdot x+\sup_{y}\{f(y)+tg(y)-p\cdot y\}\Big\}=\cav(f+tg)(x)$. 
\qed

\subsubsection{Proof of Corollary \ref{lem:actions}, upper bound $|\Omega|+1$}
Corollary \ref{lem:actions} follows from a well-known fact about concavification. 
\begin{fact}\label{caratheodory} In the optimization problem,
$$\cav f(x)=\sup\Big\{\sum_m\lambda_m f(x_m) : \sum_m\lambda_m x_m=x\Big\},$$
where $f$ is defined on $X\subseteq\reals^d$, the number of points can be restricted  to $d+1$. That is, without loss of generality, the supremum is taken over families $(\lambda_m,x_m)_{m=1}^{d+1}$.
\end{fact}
The reader is referred to \citet[Corollary 17.1.5, p. 157]{rockafellar1970convex}.
This implies that in a persuasion problem with unrestricted communication, the number of messages can be bounded by the dimension of $\Delta(\Omega)$ plus one, that is the number of states.

\begin{corollary}\label{d+2points} In the optimisation problem,
$$F^g(x,\gamma)=\sup\Big\{\sum_m\lambda_m f(x_m) : \sum_m\lambda_m x_m=x, 
\sum_m\lambda_mg(x_m)\geq \gamma\Big\}$$
where $f$ is defined on $X\subseteq\reals^d$, the number of points can be restricted  to $d+2$. \end{corollary}
\noindent This follows from Corollary \ref{coro:Cav} and Fact \ref{caratheodory}, since the function $f^g$ is defined on $X\times\reals\subseteq\reals^{d+1}$. Applying to the problem of optimal splitting under information constraint, gives a number of messages bounded by the dimension of $\Delta(\Omega)$ plus two, that is the number of states plus one. \qed

\subsubsection{Proof of Corollary \ref{lem:actions}, upper bound $|A|$}
Let $\widetilde{A}(\nu) =\argmin\bigg\{ \sum_{\omega}\nu(\omega) u_S(\omega,a) : a\in A^*(\nu)\bigg \}$ be the set of optimal actions of the receiver at $\nu$ which are worse for the sender. 

\begin{claim}
For any action $a$, the set of $\nu$'s such that $a\in\widetilde{A}(\nu)$ is convex.
\end{claim}

\proof Observe first that the set of $\nu$'s such that $a\in A^*(\nu)$ is defined by linear inequalities, i.e. the optimality of $a$, therefore is convex.
Consider now $a\in \widetilde{A}(\nu_1)\cap \widetilde{A}(\nu_2)$ and let us show that $a\in \widetilde{A}(t\nu_1+(1-t)\nu_2)$ for $t\in(0,1)$. We have $a\in A^*(\nu_1)\cap A^*(\nu_2)$ and by the remark above, $a\in A^*(t\nu_1+(1-t)\nu_2)$. Take $b \in A^*(t\nu_1+(1-t)\nu_2)$. We thus have
 $$\sum_{\omega}(t\nu_1(\omega)+(1-t)\nu_2(\omega)) u_R(\omega,a)=\sum_{\omega}(t\nu_1(\omega)+(1-t)\nu_2(\omega)) u_R(\omega,b).$$
 Since  $a\in A^*(\nu_1)\cap A^*(\nu_2)$, 
$$\sum_{\omega}\nu_1(\omega) u_R(\omega,a)\ge\sum_{\omega}\nu_1(\omega) u_R(\omega,b),\;\;\sum_{\omega}\nu_2(\omega) u_R(\omega,a)\ge\sum_{\omega}\nu_2(\omega) u_R(\omega,b).$$
Combined together, we get  $b\in A^*(\nu_1)\cap A^*(\nu_2)$. Since $a\in \widetilde{A}(\nu_1)\cap \widetilde{A}(\nu_2)$,
 $$\sum_{\omega}\nu_1(\omega) u_R(\omega,a)\le\sum_{\omega}\nu_1(\omega) u_R(\omega,b),\;\;\sum_{\omega}\nu_2(\omega) u_S(\omega,a)\le\sum_{\omega}\nu_2(\omega) u_S(\omega,b).$$
Taking the convex combination of these two inequalities proves the claim. \qed \\

Consider a feasible splitting $(\lambda_m,\mu_m)$ such that $\sum_m\lambda_m\nu_m=\mu$ and $\sum_m\lambda_mH(\nu_m)\geq H(\mu)-C$. 
For each action $a$, define $M(a)=\Big\{m : \widetilde{A}(\nu_m)=\{a\}\Big\}$. Denote $\tilde\lambda_a=\sum_{m\in M(a)}\lambda_m$ and 
$$\tilde\nu_a=\sum_{m\in M(a)}\frac{\lambda_m}{\tilde\lambda_a}\nu_m.$$
We have:
\begin{eqnarray*}
\mu&=&\sum_m\lambda_m\nu_m\\
&=&\sum_a\tilde\lambda_a\sum_{m\in M(a)}\frac{\lambda_m}{\tilde\lambda_a}\nu_m\\
&=&\sum_a\tilde\lambda_a\tilde\nu_a.\\
\end{eqnarray*}
This defines a splitting of $\mu$ with $|A|$ elements. We argue that the payoff is the same as the initial splitting.  Let us calculate the expected payoff. From the previous claim, for each action $a$, $a\in\widetilde{A}(\tilde\nu_a)$. We thus have:
\begin{eqnarray*}
\sum_m\lambda_mu^*_S(\nu_m)&=&\sum_a\tilde\lambda_a\sum_{m\in M(a)}\frac{\lambda_m}{\tilde\lambda_a}\sum_\omega\nu_m(\omega)u^*_S(\omega,a)\\
&=&\sum_a\tilde\lambda_a\sum_\omega\tilde\nu_a(\omega)u^*_S(\omega,a)\\
&=&\sum_a\tilde\lambda_au^*_S(\tilde\nu_a).\\
\end{eqnarray*}
To conclude the proof, we check that the information constraint is satisfied. This follows from the concavity of entropy. Indeed,
$$H(\tilde\nu_a)\geq\sum_{m\in M(a)}\frac{\lambda_m}{\tilde\lambda_a}H(\nu_m)$$
and thus,
$$\sum_a\tilde\lambda_a H(\tilde\nu_a)\geq\sum_m\lambda_mH(\nu_m)\geq H(\mu)-C.$$
\qed

\subsection{Proof of Theorem \ref{mainthm}, point 1, the upper bound}\label{sec:proofconverse}
{\em 1. For each pair of integers $k,n$, $U^*_S(\mu_n,Q_k)\leq V(\mu,\frac{k}{n}C(Q)).$}

\proof Let us fix a strategy $\sigma$ of the sender. This induces a probability distribution $\prob_\sigma$ of sequences in $\Omega^n\times X^k\times Y^k$, the associated random sequences are denoted $(\rv\omega^n,\rv x^k,\rv y^k)$. Let $\rv t$  be a  uniformly distributed random variable over $\{1,\ldots,n\}$, independent from $(\rv\omega^n,\rv x^k,\rv y^k)$ and denote $\rv m=(\rv y^k,\rv t)$ taking values in $M=Y^k\times \{1,\ldots,n\}$.

We denote $ \widetilde{\prob}\big( \omega,  m)$ the joint probability distribution of $(\rv\omega,\rv m)$ defined by:
\begin{align*}
 \widetilde{\prob}\big( \omega,  m) =& \widetilde{\prob}\big(\rv \omega = \omega, (\rv y^k, \rv t)=  m \big)  \\
 =& \widetilde{\prob}(\rv t=t) \cdot  \widetilde{\prob}\big( \rv \omega = \omega, \rv y^k = y^k\big|\rv t=t\big) \\
 =& \frac{1}{n} \cdot  \prob_{\sigma}\big(\rv \omega_t = \omega,  \rv y^k = y^k\big) . \label{eq:distributionW}
\end{align*}
Note that the marginal distribution of $\widetilde{\prob}\big( \omega,  m)$  on $\Omega$ is equal to the prior $\mu$:
\begin{align*}
\widetilde{\prob}\big( \omega) =& \sum_{t, y^k} \widetilde{\prob}\big( \rv \omega = \omega, \rv y^k = y^k, \rv t=t\big) \\
=& \sum_{t, y^k} \frac{1}{n} \cdot  \prob_{\sigma}\big(\rv \omega_t = \omega,  \rv y^k = y^k\big) \\
=& \sum_{t = 1}^n \frac{1}{n} \cdot  \prob_{\sigma}\big(\rv \omega_t = \omega\big)\\
=& \prob_{\sigma}\big(\omega\big) \cdot  \sum_{t = 1}^n \frac{1}{n}  = \mu(\omega).
 \end{align*}
Fix now a strategy $\tau$ of the receiver $\tau:Y^k\to A^n$ and define   $ \tilde{\tau}:M \to A$ where $\tilde{\tau}(m) = \tilde{\tau}(y^k,t) = \tau_t(y^k)$, the $t$-th coordinate of $\tau(y^k)$.
The expected average payoff of player $i = R,S$ writes:
\begin{align}
\E_{\sigma,\tau} \big[\bar{u}_i\big] =&\sum_{\omega^n,x^k,y^k}\prob_\sigma(\omega^n,x^k,y^k) \Bigg[  \frac{1}{n} \sum_{t=1}^n u_i\Big(\omega_t,\tau_t(y^k)\Big) \Bigg] \\
=&\sum_{t=1}^n \sum_{\omega_t,x^k,y^k} \frac{1}{n} \cdot   \prob_\sigma(\omega_t,x^k,y^k) \cdot u_i\Big(\omega_t,\tau_t(y^k)\Big)  \label{converse1}\\
=& \sum_{t=1}^n \sum_{\omega_t,y^k}\frac{1}{n} \cdot \prob_\sigma(\omega_t,y^k)\cdot  u_i\Big(\omega_t,\tau_t(y^k)\Big) \label{converse2} \\
=& \sum_{\omega,y^k, t} \widetilde{\prob}(\omega,y^k,t) \cdot u_i\Big(\omega,\tilde{\tau}(y^k,t)\Big)  \\
=& \sum_{\omega,m} \widetilde{\prob}(\omega,m) \cdot u_i\Big(\omega,\tilde{\tau}(m)\Big) .
\end{align}
Equation =(\ref{converse1}) implies Equation (\ref{converse2}) by summing over $x^k$ which does not enter the payoff function. All other steps are reorderings and change of variables. \\

A strategy $\tau$ is a best-reply to $\sigma$ if and only if:
\begin{align*}
&\tau(y^k)\in\arg\max_{a^n\in A^n}\sum_{\omega^n,x^k,y^k}\mu(\omega^n)\sigma(x^k|\omega^n)Q(y^k|x^k)\bar{u}_R(\omega^n,a^n)\\
\Longleftrightarrow&\quad  \tilde{\tau}(m) \in\arg\max_{a\in A}\sum_{\omega,m}\widetilde{\prob} (\omega,m) \cdot  u_R(\omega,a) \\
\Longleftrightarrow&\quad  \tilde{\tau}(m) \in\arg\max_{a\in A}\sum_{\omega}\tilde\nu_{\sigma}(\omega|m)  \cdot  u_R(\omega,a) \\
\Longleftrightarrow&\quad  \tilde{\tau}(m) \in A^*\Big(\tilde\nu_{\sigma}(\cdot|m)\Big)
\end{align*}
where $\tilde\nu_{\sigma}(\omega|m)=\widetilde{\prob}(\omega |m)$. We deduce for any strategy $\sigma$ of the sender and any best-reply $\tau$ of the sender, the expected average payoffs are those induced by the splitting:
$$\mu(\omega)=\sum_m \widetilde{\prob}(m) \tilde\nu_{\sigma}(\omega|m).$$

Now, we bound the mutual information of this splitting. Throughout the proof, we will abuse our notations in a way that is common in information theory (see \cite{cover-book-2006}). When $\rv x$ is a random variable with distribution $p(x)$, we write $H(\rv x)$ for $H(p)$, when $(\rv x,\rv y)$ is a pair of random variables with joint distribution $p(x,y)$, we write $H(\rv y|\rv x)$ for $\sum_xp(x)H(p(\cdot|x))$. Last, we will write $I(\rv x;\rv y)$ without explicit reference to the joint distribution.

For any strategy $\sigma$, we have:
\begin{align}
0 &\leq I(\rv x^k ; \rv y^k) - I(\rv\omega^n;\rv y^k) \label{eq:ConverseWaux1} \\
&= \sum_{t=1}^k H( \rv y_t|\rv y^{t-1} )  -  \sum_{t=1}^k H(\rv y_t | \rv x^k, \rv y^{t-1})-  \sum_{t=1}^n  H(\rv\omega_t|\rv\omega^{t-1}) + \sum_{t=1}^n H(\rv\omega_t|\rv y^k, \rv\omega^{t-1})  \label{eq:ConverseWaux1b} \\
&\leq \sum_{t=1}^k H(\rv y_t)  -  \sum_{t=1}^k H(\rv y_t | \rv x_t)-  n \cdot  H(\rv\omega) + \sum_{t=1}^nH(\rv\omega_t|\rv y^k)  \label{eq:ConverseWaux2} \\
&=  \sum_{t=1}^k I(\rv x_t ; \rv y_t)    -  n \cdot  H(\rv\omega) +n \cdot   \sum_{t=1}^n \prob(\rv t =t) \cdot H(\rv\omega | \rv y^k,\rv t=t)  \label{eq:ConverseWaux3} \\
&\leq k \cdot \max_{\prob(x)} I(\rv x ; \rv y)  -  n \cdot  H(\rv\omega) +n \cdot  H(\rv\omega | \rv y^k,\rv t)  \label{eq:ConverseWaux4} \\
&= k \cdot \max_{\prob(x)} I(\rv x ; \rv y)  -  n \cdot  H(\rv\omega) +n \cdot  H(\rv\omega | \rv m)  \label{eq:ConverseWaux5} \\
&= k \cdot  \max_{\prob(x)} I(\rv x ; \rv y)  -  n \cdot I(\rv\omega ; \rv m)  \label{eq:ConverseWaux6} .
\end{align}
- Equation \eqref{eq:ConverseWaux1} holds since the triple $(\rv\omega^n,\rv x^k,\rv y^k)$ has the Markov chain property; that is, its join distribution writes $\mu(\omega^n)\sigma(x^k|\omega^n)Q(y^k|x^k)$. This implies $ I(\rv x^k ; \rv y^k) \geq I(\rv\omega^n;\rv y^k)$, that is $\rv x^k$ is more informative that $\rv\omega^n$ about $\rv y^k$  \cite[Theorem 2.8.1, p. 34]{cover-book-2006}.\\
- Equation \eqref{eq:ConverseWaux1b} comes from the chain rule of entropy $ H( \rv y^k) = \sum_{t=1}^k H( \rv y_t|\rv y^{t-1} ) $.\\
- Equation \eqref{eq:ConverseWaux2} follows since the channel is memoryless $H(\rv y_t | \rv x^k, \rv y^{t-1}) = H(\rv y_t | \rv x_t)$,  the sequence of states is i.i.d.  $H(\rv\omega_t|\rv\omega^{t-1}) = H(\rv\omega_t)$, and conditioning reduces  entropy $H(\rv\omega_t|\rv y^k, \rv\omega^{t-1}) \leq H(\rv\omega_t|\rv y^k) $.\\
- Equation \eqref{eq:ConverseWaux3} is a simple rewriting with the introduction of the uniform random variable $\rv t\in\{1,\ldots,n\}$.\\
- Equation \eqref{eq:ConverseWaux4} comes from taking the maximum over the marginal distribution $\prob(x)$.\\
- Equation \eqref{eq:ConverseWaux5} comes from the change of variable $\rv m = (\rv y^k,\rv t)$.\\
Then, Equation \eqref{eq:ConverseWaux6} is equivalent to:
\begin{align*}
&\quad k \cdot  \max_{\prob(x)} I(\rv x ; \rv y)  -  n \cdot I(\rv\omega ; \rv m)  \geq 0 \\
\Longleftrightarrow& \quad  H(\rv\omega | \rv m) \geq  H(\rv\omega) - \frac{k}{n} \cdot \max_{\prob(x)} I(\rv x ; \rv y)\\
\Longleftrightarrow & \quad  \sum_m\lambda_m H(\mu_m) \geq  H(\mu) - \frac{k}{n} \cdot C(Q). 
\end{align*}
Therefore, for any strategy $\sigma$  and all $n,k$, we have:
\begin{align*}
&\min_{\tau\in BR(\sigma)}\sum_{\omega^n,x^k,y^k}\mu(\omega^n)\sigma(x^k|\omega^n)Q(y^k|x^k)\bar{u}_S(\omega^n,\tau(y^k))\\
=&\min_{\tilde{\tau}\in BR(\sigma)}\sum_{\omega,m} \widetilde{\prob}(\omega,m) \cdot  u_S\Big(\omega,\tilde{\tau}(m)\Big)\\
=& \sum_{m} \widetilde{\prob}(m) \min_{\tilde{\tau}(m) \in  A^*(\tilde\nu_{\sigma}(\cdot|m))}\sum_{\omega} \tilde\nu_{\sigma}(\cdot|m) \cdot  u_S\Big(\omega,\tilde{\tau}(m)\Big)\\
=& \sum_{m} \widetilde{\prob}(m) \cdot u^*_S\big(\tilde\nu_{\sigma}(\cdot|m)\big)\\
\leq&\sup_\sigma  \bigg\{  \sum_{m} \widetilde{\prob}(m) \cdot u^*_S\big(\tilde\nu_{\sigma}(\cdot|m)\big)\\
&\qquad \text{ s.t. }\sum_m \widetilde{\prob}(m) \cdot \tilde\nu_{\sigma}(\cdot|m) = \mu ,\\
& \qquad\text{ and } \sum_m\widetilde{\prob}(m)  \cdot H\big(\tilde\nu_{\sigma}(\cdot|m)\big) \geq  H(\mu) - \frac{k}{n} \cdot C(Q) \bigg\}\\
=&\sup  \bigg\{  \sum_{m} \lambda_m \cdot u^*_S\big(\nu_m\big)\\
& \qquad\text{ s.t. }\sum_m\lambda_m\nu_m = \mu ,\\
& \qquad\text{ and } \sum_m\lambda_m  H\big(\nu_m\big) \geq  H(\mu) - \frac{k}{n} \cdot C(Q) \bigg\}\\
&= V(\mu,\frac{k}{n}C(Q)).
\end{align*}
This proves that for all $n$ and $k$ we have:
\begin{align*}
U^*_S(\mu_n,Q_k) = & \sup_\sigma\min_{\tau\in BR(\sigma)}\sum_{\omega^n,x^k,y^k}\mu(\omega^n)\sigma(x^k|\omega^n)Q(y^k|x^k)\bar{u}_S(\omega^n,\tau(y^k))\\
\leq&  V(\mu,\frac{k}{n}C(Q)) 
\end{align*}
as desired. \qed

\subsection{Proof of Theorem \ref{mainthm}, point 2, the limit value}\label{sec:proofatteignable}

{\em 2. For each $r\in[0,+\infty]$ and each pair of sequences  $(k_j,n_j)_{j\in\naturals}$ such that $\lim\limits_{j\to\infty}\max(n_j,k_j)= \infty$ and 
 $\lim\limits_{j\to\infty}\frac{k_j}{n_j}=r$, we have $\lim_{j\to\infty}U^*_S(\mu^{n_j},Q^{k_j}) = V\big(\mu,r C(Q)\big).$}\\

For this proof we will consider the case  $r$ finite, $n_j\to \infty$, $k_j\to \infty$, $\frac{k_j}{n_j}\to r$. The proof for the case where $n_j\to \infty$, $k_j\to \infty$, $\frac{k_j}{n_j}\to \infty$ is a consequence by considering $r$ finite large enough such that $rC(Q)> H(\mu)$. For the cases where either $n_j$ or $k_j$ is bounded, see Section \ref{implications}. 
In the rest of the proof, we will consider pairs of integers $(k,n)$ which are  generic terms of such a sequence $(k_j,n_j)_{j\in\naturals}$ and omit the index $j$ for simplicity of notations.

\subsubsection{Zero capacity.}  First, we investigate the case $C(Q)=0$.

\begin{lemma}\label{lemma:ZeroCapacity}
If the channel capacity is equal to zero $\max_{p(x)} I(\rv x; \rv y) =0$, then for all $k,n$, we have:
\begin{align*}
U^*_S(\mu^n,Q^k)  =  V(\mu, \frac{k}{n} C(Q))= u^*_S(\mu).
\end{align*}
\end{lemma}

\begin{proof}[Lemma \ref{lemma:ZeroCapacity}]
Let $(\rv x, \rv y)$ be a pair of random variables such that the conditional probability of $\{\rv y=y\}$ given $\{\rv x=x\}$ is $Q(y|x)$. If the capacity of the channel is 0, then $I(\rv x;\rv y)=H(\rv y)-H(\rv y|\rv x)=0$ which implies that $\rv x$ and $\rv y$ are independent: no information can be sent through the channel.
This implies that for any splitting which satisfies the information constraint, the  random variables $\rv \omega$ and $\rv m$ are independent, and for all $m\in M$ we have $\nu_m = \mu$. Hence:
\begin{align*} 
 V(\mu, \frac{k}{n} C(Q)) = u^*_S(\mu).
\end{align*} 
Moreover, for any strategy $\sigma$, the sequence of messages $\rv y^k$ of the receiver is independent from the sequence of states $\rv\omega^n$. It follows that:
\begin{align*}
U^*_S(\mu^n,Q^k)  =&  \sup_\sigma\min_{\tau\in BR(\sigma)}\sum_{\omega^n,x^k,y^k}\mu^n(\omega^n)\sigma(x^k|\omega^n)Q^k(y^k)\bar{u}_S(\omega^n,\tau(y^k))\\
=& \min_{\tau\in BR(\sigma)}\sum_{\omega^n,y^k}\mu^n(\omega^n)Q^k(y^k) \bigg[\frac{1}{n}\sum_{t=1}^n{u}_S(\omega_t,\tau_t(y^k))\bigg]\\
=&\frac{1}{n}\sum_{t=1}^n  \min_{a_t\in A^*(\mu) }\sum_{\omega_t}\mu(\omega_t)  {u}_S(\omega_t,a_t)\\
=&  \min_{a\in A^*(\mu) }\sum_{\omega}\mu(\omega)  {u}_S(\omega,a)=u^*_S(\mu),
\end{align*}
which concludes the proof.
\qed\end{proof}

\vskip1cm

\subsubsection{Positive channel capacity.} We assume from now on $C(Q)>0$. The goal is to take a splitting of the prior which satisfies the information constraint and to show that the associated payoff can be approximately achieved by strategy $\sigma$ of the sender and a best-reply $\tau\in BR(\sigma)$ of the receiver. The next lemma states that we can focus on splittings such that the information constraint is satisfied with strict inequality and where the action of the receiver is unique for each posterior. Concretely, we prove that such splittings are dense in the set of feasible splittings. Recall that we denote $\widetilde{A}(\nu)$ the set of worst optimal actions when the belief is $\nu\in \Delta(\Omega)$:
\begin{align*}
\widetilde{A}(\nu)= \argmin\bigg\{ \sum_{\omega}\nu(\omega) u_S(\omega,a) : a\in A^*(\nu) \bigg\}.
\end{align*}
Consider the following program:
\begin{align*}
\widehat{V}(\mu,\frac{k}{n}C(Q))=&\sup \bigg\{ \sum_m\lambda_mu^*_S(\nu_m)\\
&\qquad \mathrm{ s.t.}   \sum_m\lambda_m\nu_m=\mu,\\
&\qquad \mathrm{ and } \;\;  H(\mu)-\sum_m\lambda_m H(\nu_m)< \frac{k}{n}C(Q)\\
&\qquad \mathrm{ and } \;\;  \forall m, \; \widetilde{A}(\nu_m) \text{ is a singleton } \bigg\}.
\end{align*}

\begin{lemma}\label{lemma:DenseSplittingSubset}
For all integers $(k,n)$, $\mu\in\Delta(\Omega)$ and $Q$ such that $C(Q)>0$ we have:
\begin{align}
{V}(\mu,\frac{k}{n}C(Q))  = \widehat{V}(\mu,\frac{k}{n}C(Q)) . \label{eq:DenseSplittingSubset}
\end{align}
\end{lemma}
 The proof of Lemma \ref{lemma:DenseSplittingSubset} is postponed to Section \ref{sec:ProofLemmaDenseSplittingSubset}. Then, the proof of our main result continues with two lemmas. In Lemma \ref{coro:ProbaB}, we approximate the payoff yielded by any strategy. We will see the relevance of the number of stages where the actual belief of the receiver is close to the desired target, and the importance of this number being large. Next, in Lemma \ref{coro:CodingScheme} we prove that there is a strategy for which this holds. We use there known results from information theory for defining the coding scheme which gives the strategy of the sender. Then, we prove that this strategy actually controls the Bayesian beliefs of the receiver.

Given a strategy $\sigma$ of the sender, we denote the induced expected payoff as follows:
\begin{align*}
\widehat{U}_{S, \sigma}(\mu^n,Q^k)=& \min_{\tau\in BR(\sigma)}\sum_{\omega^n,x^k,y^k}\mu(\omega^n)\sigma(x^k|\omega^n)Q(y^k|x^k)\bar{u}_S(\omega^n,\tau(y^k)),\\
=& \min_{\tau\in BR(\sigma)} \E_{\rv \omega^n,\rv x^k,\rv y^k}\big[\bar{u}_S(\rv \omega^n,\tau(\rv y^k))\big].
\end{align*}
Let $\nu^\sigma_{t,y^k}\in\Delta(\Omega)$ denote the posterior belief on $\rv \omega_t$ conditional on the sequence $y^k$. That is,
$$\nu^\sigma_{t,y^k}(\omega)= \prob_{\sigma}\big(\rv\omega_t=\omega\mid y^k\big).$$
For $\nu_1,\nu_2\in\Delta(\Omega)$, the  Kullback-Leibler (KL) divergence is,
$$D(\nu_1\| \nu_2)=\sum_\omega \nu_1(\omega)\log\frac{\nu_1(\omega)}{\nu_2(\omega)}.$$
We will introduce several positive parameters $\alpha,\gamma,\delta$, to be thought of as small. 
\begin{notation}
 For a sequence $(m^n,y^k)$ and $\alpha>0$, denote
$$T_\alpha(m^n,y^k)=\Big\{t\in\{1,\dots, n\} : D(\nu^\sigma_{t,y^k}\| \nu_{m_t})\leq \frac{\alpha^2}{2\ln 2}\Big\}.$$
\end{notation}
This is the set of indices $t=1,\dots, n$ such that the posterior belief $\nu^\sigma_{t,y^k}$ about $\rv \omega_t$ is close to the theoretical belief $\nu_{m_t}$. Intuitively, this is the set of indices where {\em the message $m_t$} is approximately transmitted. Now, we define an event $B_{\alpha, \gamma, \delta}\subseteq M^n\times Y^k$ such that for every $(m^n,y^k)\in B_{\alpha, \gamma, \delta}$, $\frac1n\sum_{t=1}^nu^*_S(\nu^\sigma_{t,y^k})$ is close to $\sum_m\lambda_mu^*_S(\nu_m)$. 
\begin{notation} For a sequence $(m^n,y^k)$ and $m\in M$, denote 
$$\freq_m(m^n,y^k)=\frac1n \Big|\{t=1,\dots,n: m_t=m\}\Big|$$
the empirical frequency of message $m$ in the sequence $m ^n$.
For $\alpha, \gamma, \delta>0$, let
$$B_{\alpha, \gamma, \delta}=\Big\{(m^n, y^k) : \frac{|T_\alpha(m^n,y^k)|}{n}\geq1-\gamma \text{ \rm and } \sum_m|\lambda_m-\freq_m(m^n,y^k)|\leq \delta\Big\}$$
\end{notation}

\begin{lemma}\label{coro:ProbaB}
\begin{align*}
\bigg|\widehat{U}_{S, \sigma}(\mu^n,Q^k)  - \widehat{V}(\mu,rC(Q))  \bigg| \leq (\alpha+2\gamma+\delta)\|u\|+(1-\prob_\sigma(B_{\alpha,\gamma,\delta}))\|u\|.
\end{align*}
\end{lemma}
The proof of Lemma  \ref{coro:ProbaB} is given in Section \ref{ProofCoroProbaB}. We see from this inequality that estimating the probability of the set $B_{\alpha,\gamma,\delta}$ is crucial and that we would like the probability of the complement $\prob_\sigma(B^c_{\alpha,\gamma,\delta})$ to be small. 

Last, Lemma \ref{coro:CodingScheme} corresponds to the actual construction of the strategy. 

\begin{lemma}\label{coro:CodingScheme}
Assume that the splitting $(\lambda_m,\nu_m)_m$ satisfies the three conditions:
\begin{align}
&\sum_m\lambda_m\nu_m=\mu,   \label{eq:Condition01}  \\
&H(\mu)-\sum_m\lambda_m H(\mu_m)< rC(Q), \label{eq:Condition02} \\
&\forall m, \; \widetilde{A}(\nu_m) \text{ is a singleton, } \label{eq:Condition03}
\end{align}
then $\forall \varepsilon>0$, $\forall \alpha >0$, $\forall \gamma >0$, $\exists \bar{\delta}$,  $\forall \delta\leq \bar{\delta}$, $\exists \bar{n}$, $\forall n\geq \bar{n}$, $\exists \sigma$, such that
$\prob_\sigma(B^c_{\alpha,\gamma,\delta})  \leq  \varepsilon$.
\end{lemma}
The proof of Lemma \ref{coro:CodingScheme} is in Appendix \ref{ProofCoroCodingScheme}. The idea is that, since the information constraint is satisfied i.e.  $I(\rv \omega;\rv m)< r C(Q)$, there is enough capacity to transmit $\rv m$ over the channel. More precisely, we construct a strategy such that the set $B_{\alpha,\gamma,\delta}$ has probability close to 1. This way, for most sequences $(\omega^n,m^n,x^k,y^k)$, the receiver {\em gets the right message} in most stages. That is, at most stages the receiver plays the action corresponding to the message.\\

We may now conclude the main proof. We combine the inequality of Lemma \ref{coro:ProbaB} with the bound $\prob_\sigma(B^c_{\alpha,\gamma,\delta})  \leq  \varepsilon$ of Lemma \ref{coro:CodingScheme}. We choose the parameters $\alpha,\gamma,\eta,\delta$ small and then $n$ large in order to obtain the following:

\begin{proposition}\label{lemma:ExistenceSplittingStrategy}
For all $r>0$ and  $\vare>0$, there exists  integers $N(\vare), K(\vare)$ such that for all $n\geq N(\vare)$,  $k\geq K(\vare)$ and $|\frac{k}{n}-r|\leq\vare$, there exists a strategy $\sigma$ such that:
\begin{align}
\bigg|\widehat{U}_{S, \sigma}(\mu^n,Q^k)  - \widehat{V}(\mu,rC(Q))  \bigg| \leq \varepsilon.
\end{align}
\end{proposition}

With Lemma \ref{lemma:DenseSplittingSubset}, this ends the proof of point 2 of Theorem \ref{mainthm}. \qed

\subsubsection{Proof of Lemma \ref{lemma:DenseSplittingSubset}}\label{sec:ProofLemmaDenseSplittingSubset}
\begin{remark}\rm  From Corollary \ref{lem:actions}, we know that we can restrict the number of messages, i.e.  the number of posteriors to $K=\min\{|A|,|\Omega|+1\}$. Therefore, from now on a splitting $(\lambda_m,\nu_m)_m$ will be understood to be a composed of $\lambda=(\lambda_1,\dots, \lambda_K)\in\Delta(\{1,\dots, K\})$ and $(\nu_m)_m\in(\Delta(\Omega))^K$. The set of splittings of $\mu$ is thus a convex and compact subset of
$\Delta(\{1,\dots, K\})\times (\Delta(\Omega))^K$
which itself is a compact and convex set in some finite dimension space. All statements below about closed or open sets of splittings relate to the topology  induced by the Euclidean topology on this finite dimension space.
\end{remark}
We consider the following sets:
\begin{align*}
\set S_1 = \bigg\{ (\lambda_m,\nu_m)_m,  & \quad \text{ s.t. }\quad  \sum_m\lambda_m\nu_m=\mu,\\
& \quad  \text{ and } \quad \sum_m\lambda_m H(\nu_m) \geq  H(\mu)- \frac{k}{n}C(Q)\bigg\},\\
\set S_2 = \bigg\{ (\lambda_m,\nu_m)_m, & \quad \text{ s.t. }  \quad  \sum_m\lambda_m\nu_m=\mu,\\
& \quad \text{ and } \quad \forall m, \; \widetilde{A}(\nu_m) \text{ is a singleton }\bigg\},\\
\set S_3 = \bigg\{ (\lambda_m,\nu_m)_m, & \quad \text{ s.t. }  \quad  \sum_m\lambda_m\nu_m=\mu,\\
& \quad \text{ and }\quad  \sum_m\lambda_m H(\nu_m) >  H(\mu)- \frac{k}{n}C(Q)\bigg\}.
\end{align*}

We will prove that the set $\set S_2\cap\set S_3$ is dense in $\set S_1$, which will imply that Equation \eqref{eq:DenseSplittingSubset} is satisfied. We first argue that $\widetilde{A}(\nu)$ is a singleton for an open and dense set of posteriors $\nu$.

\begin{definition}
Two actions $a$ and $b$ are equivalent $a\sim_i b$ for player $i=S,R$, if
for all $ \omega\in\Omega$, $u _i(\omega,a)= u_i(\omega,b).$
\end{definition}
We say that two actions $a$ and $b$ are {\em completely equivalent } if they are equivalent for both players. Without loss of generality, we assume that no two actions are completely equivalent. Otherwise, we can merge them into one single action and work on the reduced problem.

Denote $F_i \subseteq \Delta(\Omega)$ the set of beliefs for which player $i\in\{S,R\}$ is indifferent between two actions which are not equivalent:
\begin{align*}
F_i  = \bigg\{\nu \in \Delta(\Omega) : \;\; \exists a,b,\; a\nsim_i b,\sum_{\omega}\nu(\omega) u_i(\omega,a)= \sum_{\omega}\nu(\omega) u_i(\omega,b) \bigg\}.
\end{align*}
Let $F^c = \Delta(\Omega) \setminus \Big( F_R \cup F_S \Big)$ be the set of beliefs where at least one player is not indifferent between any two actions.

\begin{claim}\label{lemma:singleton} The set $F^c$ is open and dense in $\Delta(\Omega)$ and for each $\nu\in F^c$, $\widetilde{A}(\nu)$ is a singleton.
\end{claim}
\begin{proof}[Claim \ref{lemma:singleton}]
For each $i$ and each pair of actions $a,b$ with $a\nsim_i b$, the set, 
\begin{align*}
F_i(a,b)  = \bigg\{\nu \in \Delta(\Omega) : \;\; \sum_{\omega}\nu(\omega) u_i(\omega,a)= \sum_{\omega}\nu(\omega) u_i(\omega,b) \bigg\}
\end{align*}
is a closed hyperplane of dimension $\mathrm{dim}(F_i(a,b)) \leq | \Omega| -2$. Thus, $F_R$ and $F_S$ are closed and $F_R\cup F_S$ is  included in a finite union of hyperplanes of dimension at most $| \Omega| -2$. The complementary set is thus open and dense in $\Delta(\Omega)$.

Then, if $\widetilde{A}(\nu)$ contains two distinct actions $a\neq b$,  both players are indifferent between $a$ and $b$ at $\nu$. Thus, if $\nu\in F^c$,   $\widetilde{A}(\nu)$ is a singleton.
\qed
\end{proof}
It follows that $\set S_2$ is open and  dense in $\set S_1$.
\begin{claim}\label{cl:closure} If the channel capacity is strictly positive $C(Q)  >0$, the set $\set S_3$ is nonempty, open and dense in $\set S_1$.
\end{claim} 

\begin{proof}[Claim \ref{cl:closure}] Take a feasible splitting $(\lambda_m,\nu_m)_m$ in $\set S_1$:
\begin{align*}
\sum_m\lambda_m H(\nu_m) \geq H(\mu) - \frac{k}{n}C(Q).
\end{align*}
For $\vare>0$, consider the perturbed splitting $(\lambda_m,(1-\vare)\nu_m+\vare\mu)_m$. From concavity of the entropy,
\begin{align*}
\sum_m\lambda_m H((1-\vare)\nu_m+\vare\mu)&\geq(1-\vare)\sum_m\lambda_m H(\nu_m)+\vare H(\mu),\\
&\geq H(\mu) - \frac{k}{n}C(Q) + \vare \frac{k}{n}C(Q)\\
&> H(\mu) - \frac{k}{n}C(Q);
\end{align*}
thus, the information constraint is satisfied with strict inequality for $\vare>0$. It follows that $\set S_3$ is nonempty and dense in $\set S_1$. By continuity of the entropy, $\set S_3$ is open in $\set S_1$. \qed\end{proof}

Since $\set S_2$ and $\set S_3$ are open and dense,  $\set S_2\cap \set S_3$ is also open and dense in $\set S_1$. We can conclude that
${V}(\mu,\frac{k}{n}C(Q))  = \widehat{V}(\mu,\frac{k}{n}C(Q))$ as desired. This follows from the fact that  the function
$$u^*_S(\nu)=\min_{a\in A^*(\nu)}\sum_\omega \nu(\omega)u_{S}(\omega,a)$$ 
is lower-semi continuous and the supremum of an l.s.c. function over a dense set is the supremum over the full set.

It should be noticed that this is the only argument in the proof where the assumption that the receiver chooses the worst action for the sender, has a bite. When the receiver chooses the best action for the sender, we should consider $u^{**}_S(\nu)=\max_{a\in A^*(\nu)}\sum_\omega \nu(\omega)u_{S}(\omega,a)$ which is upper-semi continuous. In that case, the supremum over the dense set $\set S_2\cap \set S_3$ might be less than the supremum over $\set S_1$. However, this can only happen when the information constraint is binding at optimum and all posteriors in the optimal splitting are points of indifference for the receiver. This case is nongeneric in our class of persuasion problems: a slight change of the payoff function of the receiver would perturb the points of indifference and thus the points of discontinuity of $u^*$ and $u^{**}$. \qed

\subsubsection{Proof of Lemma \ref{coro:ProbaB}} \label{ProofCoroProbaB}

The strategy $\sigma$ induces a joint probability distribution $\prob_\sigma$ over $\Omega^n\times M^n\times X^k\times Y^k$:
\begin{align*}
 \prob_{\sigma}\big( \omega^n,  m^n,x^k,y^k) =& \prod_{t=1}^n \mu(\omega_t) \times \sigma(m^n,x^k|\omega^n) \times \prod_{t=1}^n Q(y_t|x_t).
\end{align*}

For each sequence $y^k$ of messages and for each $t$, the receiver chooses an optimal action $a_t\in A^*(\nu^\sigma_{t,y^k})$. In the worst case (for the sender), this action $a_t$ belongs to $\widetilde{A}(\nu^\sigma_{t,y^k})$. It follows that:
\begin{claim}
$$\widehat{U}_{S, \sigma}(\mu^n,Q^k)=\sum_{m^n,y^k}\prob_\sigma(m^n,y^k)\frac1n\sum_{t=1}^nu^*_S(\nu^\sigma_{t,y^k}).$$
\end{claim}

\vskip1cm

\begin{remark} Since the set of posteriors $\nu$ such that $\widetilde{A}(\nu)$ is a singleton is open, there exists $\alpha_0>0$ such that for all $m$: 
$$D(\nu\|\nu_m)\leq\alpha_0\implies \widetilde{A}(\nu)=\widetilde{A}(\nu_m).$$
\end{remark}
Whenever  $\widetilde{A}(\nu)$ is a singleton, denote  $\widetilde{A}(\nu)=\{\tilde a(\nu)\}$ the unique (worst) optimal action. From now on, we assume that $\alpha\in(0,\alpha_0)$. With the remark above, this implies that for each $t\in T_\alpha(m^n,y^k)$, the action chosen by the receiver for problem $t$ is $\tau_t(m^n,y^k)=\tilde a(\nu_{m_t})$. So precisely, $T_\alpha(m^n,y^k)$ is the set of indices $t$ such that the receiver plays the action $\tilde a(\nu_{m_t})$ which corresponds to the message $m_t$. In this sense, this is the set of indices for which the information transmission is successful.

\begin{lemma} For each $(m^n,y^k)\in B_{\alpha, \gamma, \delta}$,
$$\Big|\frac1n\sum_{t=1}^nu^*_S(\nu^\sigma_{t,y^k})-\sum_m\lambda_m u^*_S(\nu_m)\Big|\leq(\alpha+2\gamma+\delta)\|u\|,$$
where $\|u\|=\max_{\omega,a}|u_S(\omega,a)|$ is the largest absolute value of payoffs for the sender. 
\end{lemma}

\proof\, Denote $u^*=\sum_m\lambda_m u^*_S(\nu_m)$. We have:
\begin{align*}
\Big|\frac1n\sum_{t=1}^nu^*_S(\nu^\sigma_{t,y^k})-u^*\Big| &\leq  \Big|\frac1n\sum_{t\in T_\alpha(m^n,y^k)}(u^*_S(\nu^\sigma_{t,y^k})-u^*)\Big|+\Big|\frac1n\sum_{t\notin T_\alpha(m^n,y^k)}(u^*_S(\nu^\sigma_{t,y^k})-u^*)\Big|\\
&\leq  \Big|\frac1n\sum_{t\in T_\alpha(m^n,y^k)}(u^*_S(\nu^\sigma_{t,y^k})-u^*)\Big|+\gamma\|U\|
\end{align*}
Then: 
\begin{align*}
 \Big|\frac1n\sum_{t\in T_\alpha(m^n,y^k)}(u^*_S(\nu^\sigma_{t,y^k})-u^*)\Big|&\leq
\Big|\frac1n\sum_{t\in T_\alpha(m^n,y^k)}(u^*_S(\nu^\sigma_{t,y^k})-u^*_S(\nu_{m_t}))\Big|\\
&+\Big|\frac1n\sum_{t\in T_\alpha(m^n,y^k)}(u^*_S(\nu_{m_t})-u^*)\Big|
\end{align*}
Since $\alpha\leq\alpha_0$, for each $t\in  T_\alpha(m^n,y^k)$, $\tilde a(\nu^\sigma_{t,y^k})=\tilde a(\nu_{m_t})$. Therefore, for $t\in  T_\alpha(m^n,y^k)$
$$\Big|u^*_S(\nu^\sigma_{t,y^k})-u^*_S(\nu_{m_t})\Big|\leq\sum_\omega |\nu^\sigma_{t,y^k}(\omega)-\nu_{m_t}(\omega)|\cdot|u_S(\omega,a)|\leq \| \nu^\sigma_{t,y^k}-\nu_{m_t}\|\cdot \|u\|\leq \alpha\|u\|,$$
where the latter inequality comes from Pinsker's inequality\footnote{\citealp[Lemma 11.6.1, p. 370]{cover-book-2006}.}: $\|\nu_1-\nu_2\|\leq\sqrt{2\ln 2\, D(\nu_1\| \nu_2)}$ and the definition of $ T_\alpha(m^n,y^k)$. It follows:
$$ \Big|\frac1n\sum_{t\in T_\alpha(m^n,y^k)}(u^*_S(\nu^\sigma_{t,y^k})-u^*)\Big|\leq \alpha\|u\|+\Big|\frac1n\sum_{t\in T_\alpha(m^n,y^k)}(u^*_S(\nu_{m_t})-u^*)\Big|$$
Now from $\frac{|T_\alpha(m^n,y^k)|}{n}\geq1-\gamma$, we have:
$$\Big|\frac1n\sum_{t\in T_\alpha(m^n,y^k)}(u^*_S(\nu_{m_t})-u^*)\Big|\leq \Big|\frac1n\sum_{t=1}^n(u^*_S(\nu_{m_t})-u^*)\Big|+\gamma\|u\|.$$
Then:
\begin{align*}
\Big|\frac1n\sum_{t=1}^n(u^*_S(\nu_{m_t})-u^*)\Big|&=\Big|\sum_m(\freq_m(m^n,y^k)-\lambda_m)u^*_S(\nu_m)\Big|\\
&\leq \sum_m\Big|\freq_m(m^n,y^k)-\lambda_m\Big|\cdot\Big|u^*_S(\nu_m)\Big|\\
&\leq \|u\|\delta.
\end{align*}
Collecting all inequalities together yields the desired conclusion.\qed

\subsubsection{Proof of Lemma \ref{coro:CodingScheme}} \label{ProofCoroCodingScheme}

By hypothesis, the splitting $(\lambda_m,\nu_m)_m$ satisfies the three conditions:
\begin{align}
&\sum_m\lambda_m\nu_m=\mu,   \label{eq:Condition1}  \\
&H(\mu)-\sum_m\lambda_m H(\mu_m)< rC(Q), \label{eq:Condition2} \\
&\forall m, \; \widetilde{A}(\nu_m) \text{ is a singleton. } \label{eq:Condition3}
\end{align}
Let $M=\{1,\dots, |M|\}$ be the set of messages associated with this splitting. 

\noindent{\em Part 1. Coding scheme.} We turn now to the actual construction. We use standard information theoretic techniques for Channel Coding \cite[Chap. 3.1, p. 38]{ElGammalKim(book)11} and Lossy Source Coding \cite[Chap. 3.6, p. 56]{ElGammalKim(book)11}. Using information theoretic language, the sender is viewed as an {\em encoder} who encrypts his intended $m^n$ messages in sequences of inputs $x^k$.  The messages $m^n$ are immaterial and can be seen as a pure mental construct of the sender. The encoding is such that a {\em decoder} who reads the sequence $y^k$, is able to determine the correct $m^n$ with high probability. This is described as follows.

For $\delta>0$, we define the set of {\em typical sequences}  $A_{\delta}$ as follows:
\begin{align}
A_{\delta}=  \bigg\{(\omega^n,m^n,x^k,y^k),\quad \text{ s.t. }
& \quad \sum_{\omega,m}\Big| \lambda_m \mu_m(\omega) - \freq_{\omega,m}(\omega^n,m^n) \Big| \leq \delta, \label{eq:defA1}\\
 \text{ and }&\quad \sum_{x,y}\Big| \prob(x)\times Q(y|x) - \freq_{x,y}(x^k,y^k) \Big| \leq \delta \bigg\} . \label{eq:defA2}
\end{align}
A pair of sequences $(\omega^n,m^n)$ which satisfies Equation \eqref{eq:defA1} will be called {\em jointly typical}. Similarly, pair of sequences $(x^k,y^k)$ which satisfies Equation \eqref{eq:defA2} will be called {\em jointly typical}. With a slight abuse of notation, we will write $(\omega^n,m^n)\in A_\delta$ or $(x^k,y^k)\in A_\delta$ to indicate jointly typical sequences.

Since condition  \eqref{eq:Condition2} is satisfied with  strict inequality, there exists a small parameter $\eta>0$ and a ``rate''  $\textsf{R}\geq 0 $, such that:
\begin{align}
\textsf{R}  =&      H(\mu) - \sum_m \lambda_m H(\mu_m) + \eta  \label{eq:AchievabilityBB1} , \\
\textsf{R}  \leq&   r C(Q)  -  \eta  \label{eq:AchievabilityBB2}  .
\end{align}
Moreover, we can assume that $n\textsf{R}$ is an integer for $n$ large enough.

\begin{itemize}
\item[$\bullet$] \textit{Random codebook.} A {\em codebook} is a family 
 $\mathsf{b}$ of  $|J|= 2^{n   \sf{R}   } $ sequences $m^n(j)$ and $x^k(j)$ indexed by  $j\in J $. A {\em random codebook} is the  draw of a codebook from the marginal i.i.d. probability distributions $(\lambda_m)^{\otimes n} $ and $\prob(x)^{\otimes n} $. The selected codebook is known by the encoder and the decoder. 

\item[$\bullet$] \textit{Encoding function.} The encoder observes the sequence of states $\omega^n \in  \Omega^n$. It finds an index $j\in J$ such that the sequences  $(\omega^n,m^n(j))\in A_\delta$ are jointly typical, i.e. satisfy Equation \eqref{eq:defA1}. The encoder sends the sequence $x^k(j)$ corresponding to the index $j\in J$.
\item[$\bullet$] \textit{Decoding function.} The decoder observes the sequence of channel output $y^k\in Y^k$. It finds an index $\hat{j}\in J$ such that the sequences  $\big(x^k(\hat{j}), y^k\big)\in A_\delta$ are jointly typical, i.e. satisfy Equation \eqref{eq:defA2}. The decoder decodes the sequence $m^n(\hat{j})$.  

\item[$\bullet$] \textit{Error Event.} We introduce the indicator of error $E_{\delta} \in \{0,1\}$ defined as follows:
\begin{align}
 E_{\delta} = \Bigg\{
\begin{array}{lll}
0 \text{ if }&  j= \hat{j}  \;\; \text{ and }\;\;  \big(\omega^n ,  m^n, x^k,y^k \big)    \in A_{\delta} ,\\
1 \text{ if }&  j \neq \hat{j} \;\; \text{ or } \;\; \big(\omega^n ,  m^n, x^k,y^k \big)    \notin A_{\delta} .
\end{array}
\Bigg.
\end{align}
An error $E_{\delta}=1$ occurs in the coding process if: 1) the indices $j\in J$ and $\hat{j}\in J$ are not equal or 2) the sequences of symbols $\big( \omega^n ,m^n , x^k,y^k\big) \notin A_{\delta}$, i.e. are not jointly typical. 
\end{itemize}

An important result in information theory is that the expected probability of error over the random codebook is small.

\textit{Expected error probability.} For all $\varepsilon_2>0$, for all $\eta >0$, there exists a $\bar{\delta}>0$, for all $\delta \leq \bar{\delta}$  there exists $\bar{n}, \bar k$ such that for all $n\geq\bar{n}, k\geq \bar k$ and $|\frac{k}{n}-r|\leq \vare_2$, the expected probability of the following error events are bounded by $\varepsilon_2$:
\begin{align}
&\E\bigg[ \prob_{\mathsf{b}}\bigg( \forall  j \in J  ,\quad 
\big(\omega^n, m^n(j) \big) \notin A_{\delta} \bigg)\bigg]  \leq \varepsilon_2, \label{eq:AchievProbaBB1}
 \\
&\E\bigg[ \prob_{\mathsf{b}}\bigg(  \exists j'\neq  j  ,\text{ s.t. } 
\big(y^k , x^k(j') \big) \in A_{\delta}\bigg)\bigg]   \leq \varepsilon_2. \label{eq:AchievProbaBB2}
\end{align}

\noindent - Equation \eqref{eq:AchievProbaBB1} comes from Equation \eqref{eq:AchievabilityBB1} and the Covering Lemma \ref{lemma:covering},  \cite[Lemma 3.3, p. 62]{ElGammalKim(book)11}.\\
- Equation \eqref{eq:AchievProbaBB2} comes from Equation \eqref{eq:AchievabilityBB2} and the Packing Lemma \ref{lemma:packing},  \cite[Lemma 3.1, p. 46]{ElGammalKim(book)11}.\\

If the expected probability of error is small over the codebooks, then it has to be small for at least one codebook. Following a standard analysis of the error probability,  \cite[pp. 42--43, 60--61]{ElGammalKim(book)11}, Equations \eqref{eq:AchievProbaBB1}, \eqref{eq:AchievProbaBB2}  imply that:
\begin{align}
& \forall \varepsilon_2>0,\;  \forall \eta>0, \;  \exists \bar{\delta}>0,\;\forall \delta\leq \bar{\delta}, \; \exists \bar{n}, \bar k,\;\forall n\geq \bar{n},\forall k\geq \bar k, |\frac{k}{n}-r|\leq\vare_2, \hskip5pt
\exists \mathsf{b}^\star,\text{ s.t. }  \prob_{\mathsf{b}^\star}\big(E_{\delta}=1 \big) \leq \varepsilon_2. \label{eq:BoundError0}
\end{align}

The strategy $\sigma$ of the sender consists in using this codebook $\mathsf{b}^{\star}$ in order to find the sequence $m^n(j)$ which is jointly typical with $\omega^n$, and in sending the sequence $x^k(j)$. By construction, this satisfies Equation \eqref{eq:BoundError0}, i.e. it has a low  probability of error.\\

 \noindent{\em Part 2. Control of the Beliefs.} 
The previous construction has the property that the decoder who uses the decoding schemes, makes an error with small probability. Now, the receiver needs not use the decoding scheme. Actually, the receiver calculates the posterior belief on the sequence of states $\omega^n$, given $y^k$.  Our  contribution is to show that those beliefs are close to the prescribed beliefs $\nu_m$ at most stages. We have the following chain of inequalities:
\begin{align}
&  \E_{\sigma} \Bigg[ \frac{1}{n}  \sum_{t=1}^n D\Big( \nu^\sigma_{t,y^k} \Big\|   \nu_{m_t} \Big)\, \Big| \, E_{\delta}=0 \Bigg] \nonumber \\
 =& \sum_{m^n,y^k}\prob_{\sigma}(m^n,y^k|E_{\delta}=0) \cdot \frac{1}{n}  \sum_{t=1}^n D\Big( \nu^\sigma_{t,y^k} \Big\|  \nu_{m_t} \Big) \label{eq:KLdiv0} \displaybreak[0] \\
=&\frac{1}{n}   \sum_{(\omega^n,m^n,y^k)\in A_{\delta} }\prob_{\sigma}(\omega^n,m^n,y^k|E_{\delta}=0)   \cdot \log_2 \frac{1}{\prod_{t=1}^n \nu_{m_t}(\omega_t) }  -   \frac{1}{n}  \sum_{t=1}^n H(\rv \omega_t|\rv y^k,E_{\delta}=0) \label{eq:KLdiv1} \\
\leq&\frac{1}{n}   \sum_{(\omega^n,m^n,y^k)\in A_{\delta} }\prob_{\sigma}(\omega^n,m^n,y^k|E_{\delta}=0)   \cdot \log_2 \frac{1}{\prod_{t=1}^n \nu_{m_t}(\omega_t) }  -   \frac{1}{n}  \sum_{t=1}^n H(\rv \omega_t|\rv m^n,\rv y^k,E_{\delta}=0) \label{eq:KLdiv2} \\
\leq&\frac{1}{n}   \sum_{(\omega^n,m^n,y^k)\in A_{\delta} }\prob_{\sigma}(\omega^n,m^n,y^k|E_{\delta}=0)   \cdot n \cdot \Big(H(\rv \omega|\rv m) + \delta  \Big) - \frac{1}{n}  H(\rv \omega^n|\rv m^n,\rv y^k,E_{\delta}=0) \label{eq:KLdiv3} \\
\leq&\frac1n I(\rv \omega^n;\rv m^n, \rv y^k|E_{\delta}=0)  -  I(\rv \omega;\rv m) +  \delta + \frac{1}{n}   + \log_2 |\Omega|\cdot \prob_{\sigma}\big(E_{\delta}=1 \big) \label{eq:KLdiv4}\\
\leq&\frac1n I(\rv \omega^n; \rv m^n|E_{\delta}=0)  -  I(\rv \omega;\rv m) +  \delta + \frac{2}{n}   + 2\log_2 |\Omega| \cdot \prob_{\sigma}\big(E_{\delta}=1 \big) \label{eq:KLdiv5} \\
\leq& \eta +  \delta + \frac{2}{n}   +2 \log_2 |\Omega| \cdot \prob_{\sigma}\big(E_{\delta}=1 \big)   . \label{eq:KLdiv6}
\end{align}
- Equation \eqref{eq:KLdiv0} comes from the definition of the expected K-L divergence.\\
- Equation \eqref{eq:KLdiv1} comes from the conditioning by $E_{\delta}=0$, since the support of $\prob_{\sigma}(\omega^n,m^n,y^k|E_{\delta}=0)$ is included in $A_{\delta}$.\\
- Equation \eqref{eq:KLdiv2} comes from the  property of the entropy $H(\rv \omega_t|\rv m^n,\rv y^k,E_{\delta}=0) \leq H(\rv \omega_t| \rv y^k,E_{\delta}=0) $.\\
- Equation \eqref{eq:KLdiv3} comes from the property of typical sequences $(\omega^n,m^n)\in A_{\delta}$, stated in Lemma \ref{lemma:typicalSeq} and in \citet[Property 1, pp. 26]{ElGammalKim(book)11}, and the chain rule for entropy: $$H(\rv \omega^n|\rv m^n,\rv y^k,E_{\delta}=0)   \leq \sum_{t=1}^n  H(\rv \omega_t|\rv m^n,\rv y^k,E_{\delta}=0). $$\\
- Equation \eqref{eq:KLdiv4} comes from Lemma \ref{lemma:IID} (see section \ref{sec:CoveringLemma}), which implies
$$\frac1n H(\rv \omega^n| E_{\delta}=0) - \frac1nH(\rv \omega^n) +\frac1n +  \log_2 |\Omega| \cdot \prob_{\sigma}(E_{\delta}=1)\ge0.$$
Adding this expression to Equation \eqref{eq:KLdiv3} yields Equation \eqref{eq:KLdiv4}.

- Equation \eqref{eq:KLdiv5} comes from Lemma \ref{lemma:IID} (see section \ref{sec:CoveringLemma}) which implies that $$I(\rv \omega^n;\rv  y^k|\rv m^n,E_{\delta}=0) \leq I(\rv  \omega^n;\rv y^k|\rv m^n) + 1 + n \cdot \log_2 |\Omega| \cdot \prob_{\sigma}(E_{\delta}=1) = 1 + n \cdot \log_2 |\Omega| \cdot \prob_{\sigma}(E_{\delta}=1),$$ where $I(\rv \omega^n;\rv y^k|\rv m^n)=0$,  from the Markov chain property of  the triple $(\rv\omega^n, \rv m^n,\rv y^k)$.\\
- Equation \eqref{eq:KLdiv6} comes from the cardinality of the codebook\footnote{The last argument is inspired by \citealp[Equation (23)]{MerhavShamai(StateMasking)07},  for the problem of ``Information Rates Subject to State Masking''.}: 
$$I(\rv \omega^n;\rv m^n|E_{\delta}=0)\leq H(\rv m^n)\leq \log_2|J| = n   \cdot \textsf{R}    = n \cdot  (  I( \rv \omega ; \rv  m )  + \eta ).$$

Then, we have:
\begin{align}
&1 - \prob_\sigma(B_{\alpha,\gamma,\delta}) := \prob_\sigma(B^c_{\alpha,\gamma,\delta}) \nonumber\\
=&\prob_\sigma(E_{\delta}=1) \prob_\sigma(B^c_{\alpha,\gamma,\delta}| E_{\delta}=1)  + \prob_\sigma(E_{\delta}=0) \prob_\sigma(B^c_{\alpha,\gamma,\delta}| E_{\delta}=0) \nonumber\\
\leq&\prob_\sigma(E_{\delta}=1)   +  \prob_\sigma(B^c_{\alpha,\gamma,\delta}| E_{\delta}=0) \nonumber\\
\leq&\varepsilon_2  +  \prob_\sigma(B^c_{\alpha,\gamma,\delta}| E_{\delta}=0) .\label{eq:ErrorTerm2}
\end{align}
Moreover:
\begin{align}
& \prob_\sigma(B^c_{\alpha,\gamma,\delta}| E_{\delta}=0)\nonumber\\
=&\sum_{m^n,y^k}\prob_{\sigma}\Big( (m^n,y^k)\in  B^c_{\alpha,\gamma,\delta} \Big| E_{\delta}=0\Big)  \label{eq:MarkovIneqB0} \\
=&\sum_{m^n,y^k}\prob_{\sigma}\Bigg( (m^n,y^k)\,\quad \text{ s.t. } \quad  \frac{|T_\alpha(m^n,y^k)|}{n}< 1-\gamma  \Bigg| E_{\delta}=0\Bigg)  \label{eq:MarkovIneqB1} \\
=& \prob_{\sigma}\Bigg( \frac{\#}{n} \bigg\{t , \text{ s.t. } D\bigg( \nu^\sigma_{t,y^k} \bigg\|   \nu_{m_t}  \bigg)\leq\frac{\alpha^2}{2\ln 2}    \bigg\} < 1 -\gamma \Bigg| E_{\delta}=0 \Bigg) \label{eq:MarkovIneqB2}\\
=& \prob_{\sigma}\Bigg( \frac{\#}{n} \bigg\{t , \text{ s.t. } D\bigg( \nu^\sigma_{t,y^k} \bigg\|    \nu_{m_t}  \bigg)> \frac{\alpha^2}{2\ln 2}    \bigg\} \geq \gamma \Bigg| E_{\delta}=0 \Bigg) \label{eq:MarkovIneqB2b}\\
\leq& \frac{2\ln 2}{\alpha^2\gamma}   \cdot \E_{\sigma}\Bigg[  \frac{1}{n}  \sum_{t=1}^n   D\bigg( \nu^\sigma_{t,y^k}  \bigg\|   \nu_{m_t} \bigg)\Bigg] \label{eq:MarkovIneqB3}  \\
\leq&\frac{2\ln 2}{\alpha^2\gamma}   \cdot \bigg( \eta +  \delta + \frac{2}{n}   +2 \log_2 |\Omega| \cdot \prob_{\sigma}\big(E_{\delta}=1 \big)  \bigg) \label{eq:MarkovIneqB4}.
\end{align}
- Equations \eqref{eq:MarkovIneqB0} to 
\eqref{eq:MarkovIneqB2b} are simple reformulations.\\
- Equation \eqref{eq:MarkovIneqB3} comes from a use of Markov's inequality, detailed in Lemma \ref{lemma:MarkovInequalityZ} (see section \ref{sec:CoveringLemma}).\\
- Equation \eqref{eq:MarkovIneqB4}  comes from equation  \eqref{eq:KLdiv6}.\\

\noindent Combining equations \eqref{eq:BoundError0}, \eqref{eq:ErrorTerm2}, and \eqref{eq:MarkovIneqB4} we obtain the following statement:

$\forall \varepsilon_3>0$, $\forall \alpha >0$, $\forall \gamma >0$, $ \exists \bar{\eta}$,  $\forall \eta\leq \bar{\eta}$, $\exists \bar{\delta}$,  $\forall \delta\leq \bar{\delta}$, $\exists \bar{n}, \bar k$, $\forall n\geq \bar{n},\forall k\geq \bar k$, $|\frac{k}{n}-r|\leq\vare_3$, $\exists \sigma$, such that:
$$\prob_\sigma(B^c_{\alpha,\gamma,\delta})  \leq 2 \cdot  \prob_{\sigma}\big(E_{\delta}=1 \big) + \frac{2\ln 2}{\alpha^2\gamma}  \cdot \bigg(\eta+  \delta + \frac{2}{n}   +2 \log_2 |\Omega| \cdot \prob_{\sigma}\big(E_{\delta}=1 \big) \bigg) \leq \varepsilon_3.$$
 By choosing appropriately the ``rate''  $\textsf{R}\geq 0 $ in \eqref{eq:AchievabilityBB1} and \eqref{eq:AchievabilityBB2} such as to make $\eta>0$ small, we obtain the desired result:
\begin{align*}
\forall \varepsilon>0, \;\; \forall \alpha >0, \;\; \forall \gamma >0, \;\; \exists \bar{\delta}, \;\; \forall \delta\leq \bar{\delta}, \;\; \exists \bar{n}, \bar k,\;\forall n\geq \bar{n},\forall k\geq \bar k, |\frac{k}{n}-r|\leq\vare, \;\; \exists \sigma,
\end{align*}

such that $\prob_\sigma(B^c_{\alpha,\gamma,\delta})  \leq \varepsilon$.\qed

\subsection{Additional lemmas}\label{sec:CoveringLemma}
The next three lemmas are standard results in information theory. They are recalled for the convenience of the reader.

\begin{lemma}{\rm(Covering lemma: compression of information source, Lemma 3.3, p. 62 in \citealp{ElGammalKim(book)11})}\label{lemma:covering}

Consider a random sequence $\omega^n$  with i.i.d. distribution $\prob^{\otimes n}(\omega)$ and a family of $2^{n\textsf{R}}$ sequences $\big(m^n(j)\big)_{j\in\{1,\ldots,2^{n\textsf{R}}\}}$ independently drawn from the i.i.d. distribution $\prob^{\otimes n}(m)$. Assume that $\textsf{R}= I(\rv \omega;\rv m)+\eta$ with $\eta>0$.
 
For  all $\varepsilon>0$,  there exists  $\bar{\delta}>0$, such that for all $\delta \leq \bar{\delta}$,  there exists $\bar{n}$, such that for all $n\geq\bar{n}$:
\begin{align*}
\prob\bigg( \forall  j \in J  ,\quad \big( \omega^n , m^n(j)\big) \notin A_{\delta} \bigg) \leq \varepsilon. \label{eq:covering} 
\end{align*}
 \end{lemma}

 \begin{lemma}{\rm(Packing lemma: transmission over a noisy channel, Lemma 3.1, p. 46 \citealp{ElGammalKim(book)11})}\label{lemma:packing}
 
Consider a random sequence $y^k$ drawn with i.i.d. distribution $\prob^{\otimes k}(y)$ and a family of $2^{k\textsf{R}}$ sequences $\big(x^k(j)\big)_{j\in\{1,\ldots,2^{k\textsf{R}}\}}$ independently drawn from the i.i.d. distribution $\prob^{\otimes k}(x)$.  Assume that  $\textsf{R}= I(\rv x;\rv y)-\eta$ with $\eta>0$.

For  all $\varepsilon>0$, there exists  $\bar{\delta}>0$, such that for all $\delta \leq \bar{\delta}$,  there exists $\bar{k}$, such that for all $k\geq\bar{k}$:
\begin{align*}
\prob\bigg( \exists  j  \in J  ,\quad \big(x^k(j), y^k \big) \in A_{\delta} \bigg)  \leq \varepsilon. 
\end{align*}
  \end{lemma}

\begin{lemma}[Typical sequences, Property 1, p. 26 in \citealp{ElGammalKim(book)11}]\label{lemma:typicalSeq}
The typical sequences $( \omega^n,m^n)\in A_{\delta}$ satisfy:
\begin{align*}
\forall \delta_2>0,\; \exists \bar{\delta}_2>0,\; \forall \delta \leq \bar{\delta}_2,\; \forall n,\; \forall ( \omega^n,m^n) \in A_{\delta},\\\qquad \bigg| \frac{1}{n} \cdot \log_2 \frac{1}{\prod_{t=1}^n \prob( \omega_t|m_t)} - H(\rv  \omega|\rv m) \bigg| \leq \delta_2,
\end{align*}
where $\bar{\delta}_2 = \delta_2 \cdot H(\rv  \omega | \rv m)$. 
\end{lemma}

The next two lemmas are easy ancillary  results that were used in the proofs and were omitted in the previous section to ease the reading.

\begin{lemma}[Markov's inequality]\label{lemma:MarkovInequalityZ}
For all $\varepsilon_1>0$, $\varepsilon_2>0$ we have:
\begin{align}
&\E_{\sigma} \Bigg[ \frac{1}{n}  \sum_{t=1}^n   D\bigg(  \prob_{\sigma}(\rv \omega_t|\rv y^n,E_{\delta}=0) \bigg\|   \prob(\rv \omega_t|\rv m_t) \bigg)\Bigg]   \leq  \varepsilon_0\\
\Longrightarrow& \prob_{m^n,y^n}\Bigg( \frac{\#}{n} \bigg\{t , \text{ s.t. } D\bigg(  \prob_{\sigma}(\rv \omega_t| \rv y^n,E_{\delta}=0) \bigg\|   \prob(\rv \omega_t|\rv m_t) \bigg)>   \varepsilon_1 \bigg\} > \varepsilon_2 \Bigg) \leq  \frac{\varepsilon_0}{\varepsilon_1 \cdot \varepsilon_2}.
\end{align}
\end{lemma}

\begin{proof}[Lemma \ref{lemma:MarkovInequalityZ}]
We denote by $D_t =  D\big(  \prob_{\sigma}(\rv \omega_t| \rv y^n,E_{\delta}=0) \big\|  \prob(\rv \omega_t|\rv m_t) \big)$ and $D^n = \{D_t\}_{t}$ the K-L divergence. We have that: 
\begin{align}
 \prob\Bigg( \frac{\#}{n} \bigg\{t , \text{ s.t. } D_t >   \varepsilon_1\bigg\} > \varepsilon_2 \Bigg)
=& \prob\Bigg( \frac{1}{n} \cdot \sum_{t=1}^n \UN\bigg\{D_t>   {\varepsilon_1} \bigg\} > {\varepsilon_2}\Bigg) \label{eq:twiceMarkov1}\\
\leq& \frac{ \E\bigg[ \frac{1}{n} \cdot \sum_{t=1}^n \UN\Big\{D_t>   {\varepsilon_1} \Big\}\bigg]}{\varepsilon_2} \label{eq:twiceMarkov2}\\
=& \frac{ \frac{1}{n} \cdot \sum_{t=1}^n  \E\bigg[ \UN\Big\{D_t>   {\varepsilon_1} \Big\}\bigg]}{\varepsilon_2} \label{eq:twiceMarkov3}\\
=& \frac{ \frac{1}{n} \cdot \sum_{t=1}^n  \prob\Big( D_t>   {\varepsilon_1}\Big)}{\varepsilon_2} \label{eq:twiceMarkov4}\\
\leq& \frac{ \frac{1}{n} \cdot \sum_{t=1}^n  \frac{\E[ D_t]}{\varepsilon_1}}{\varepsilon_2} \label{eq:twiceMarkov5} \\
=& \frac{1}{\varepsilon_1 \cdot \varepsilon_2} \cdot \E\bigg[\frac{1}{n} \cdot \sum_{t=1}^n  D_t\bigg] \leq \frac{\varepsilon_0}{\varepsilon_1 \cdot \varepsilon_2}. \label{eq:twiceMarkov6}
\end{align}
- Equations \eqref{eq:twiceMarkov1},  \eqref{eq:twiceMarkov3},  \eqref{eq:twiceMarkov4},  \eqref{eq:twiceMarkov6} are reformulations of probabilities and expectations.\\
- Equations \eqref{eq:twiceMarkov2},  \eqref{eq:twiceMarkov5}, come from Markov's inequality $\prob(X\geq\alpha)\leq \E[X]/\alpha$.
\qed\end{proof}



\begin{lemma}\label{lemma:IID}
Consider an i.i.d. random sequence $\rv \omega^n$. For all $\varepsilon>0$, there exists  $\bar{n}\in\N$ such that for all $n\geq\bar{n}$ we have:
\begin{align}
H(\rv \omega^n| E_{\delta}=0) \geq &\, n \cdot \bigg( H(\rv \omega) - \varepsilon \bigg).
\end{align}
\end{lemma}

\begin{proof}[Lemma \ref{lemma:IID}]
\begin{align}
H(\rv \omega^n| E_{\delta}=0) 
=&\frac{1}{\prob( E_{\delta}=0)} \cdot \bigg(  H(\rv \omega^n|  E_{\delta}=1) -  \prob(E_{\delta}=1) \cdot H(\rv \omega^n| E_{\delta}=1) \bigg)\label{eq:LemmaIID2}\\
\geq& H(\rv \omega^n |  E_{\delta}) -  \prob(E_{\delta}=1) \cdot H(\rv \omega^n| E_{\delta}=1)\label{eq:LemmaIID3}\\
\geq& H(\rv \omega^n) - H(E_{\delta}) -  \prob(E_{\delta}=1) \cdot H(\rv \omega^n| E_{\delta}=1)\label{eq:LemmaIID4}\\
\geq& H(\rv \omega^n) - n \cdot  \varepsilon.\label{eq:LemmaIID5}
\end{align}
- Equation \eqref{eq:LemmaIID2} comes from the definition of the conditional entropy.\\
- Equation \eqref{eq:LemmaIID3} comes from the property $\prob(E_{\delta}=0)\leq1$.\\
- Equation \eqref{eq:LemmaIID4} comes from the property $H(\rv \omega^n |  E_{\delta})  = H(\rv \omega^n, E_{\delta}) - H(E_{\delta}) \geq  H(\rv \omega^n) - H(E_{\delta}) $.\\
- Equation \eqref{eq:LemmaIID5} comes from the i.i.d. property of the state  $\omega$ and the definition of the error event $E_{\delta}=1$. Hence, for all $\varepsilon$, there exists a $\bar{n}\in\N$ such that for all $n\geq \bar{n}$ we have: $H(\prob(E_{\delta}=1)) +  \prob(E_{\delta}=1) \cdot \log_2|\Omega| \leq \varepsilon$.
\qed\end{proof}

\end{document}